\documentclass[12pt]{article}
\usepackage{amssymb}
\usepackage{amsfonts}
\usepackage[fleqn]{amsmath}
\usepackage{geometry}
\usepackage{graphicx}
\usepackage{subfigure}
\usepackage{natbib}
\usepackage{abraces}
\usepackage{mathtools}
\usepackage{subfigure}
\usepackage{float}
\usepackage{enumitem}
\usepackage[font=scriptsize,justification=centering]{caption}
\usepackage{setspace}
\usepackage{hyperref}
\usepackage{bigints}
\usepackage[usenames, dvipsnames]{color}
\usepackage{ed}      	
\hypersetup{ pdftitle= Sorting and Team Formation, pdfauthor=Job Boerma and Aleh Tsyvinski and Alexander Zimin,
colorlinks, citecolor=black,
filecolor=black,                                                linkcolor=Blue,                                                                                                pdftitle=title,                                                pdfauthor=author,                                                bookmarks,
citecolor= Blue,                                                pdftex}

\usepackage{amsthm}
\usepackage{mathrsfs}
\theoremstyle{definition}
\newtheorem{theorem}{Theorem}
\newtheorem*{theorem*}{Theorem}

\newtheorem{assumption}{Assumption}

\newtheorem*{corollary}{Corollary}

\newtheorem{definition}{Definition}

\newtheorem{lemma}{Lemma}

\newtheorem{proposition}[theorem]{Proposition}

\usepackage[noabbrev]{cleveref}

\begin{document}



\title{Sorting with Teams\thanks{%
\baselineskip12.5pt We thank Anmol Bhandari, Carter Braxton, Hector Chade, Jan Eeckhout, Fatih Guvenen, Kyle Herkenhoff, Rasmus Lentz, Ilse Lindenlaub, Jeremy Lise, Paolo Martellini, Simon Mongey, Chris Moser, Guillaume Sublet, and Ruodu Wang for useful discussion.}}
\author{ Job Boerma \\
{\small { \hspace{2 cm} University of Wisconsin-Madison \hspace{2 cm} } }\\  \vspace{-0.5 cm}
\and  Aleh Tsyvinski \\
{\small { \hspace{4 cm}  Yale University\hspace{4 cm} } }\\  \vspace{-0.5 cm}
\and  Alexander P. Zimin \\
{\small { \hspace{4 cm}  MIT and HSE \hspace{4 cm} } }\\
}
\date{\vspace{0.4 cm} November 2023 \vspace{-0.2 cm} }
\maketitle

\begin{abstract}
\fontsize{12.0pt}{18.0pt} \selectfont


\noindent
We fully solve a sorting problem with heterogeneous firms and multiple heterogeneous workers whose skills are imperfect substitutes. We show that optimal sorting, which we call mixed and countermonotonic, is comprised of two regions. In the first region, mediocre firms sort with mediocre workers and coworkers such that the output losses are equal across all these teams (mixing). In the second region, a high skill worker sorts with low skill coworkers and a high productivity firm (countermonotonicity). We characterize the equilibrium wages and firm values. Quantitatively, our model can generate the dispersion of earnings within and across US firms.

\vspace{0.6cm} \noindent {\bf JEL-Codes}: J01, D31, C78. 
\end{abstract}

\renewcommand{\thefootnote}{\fnsymbol{footnote}} \renewcommand{%
\thefootnote}{\arabic{footnote}} 
\thispagestyle{empty} \setcounter{page}{0}

\fontsize{12.3pt}{22.0pt} \selectfont
\newpage

\section{Introduction}

How do heterogeneous workers sort to work in teams at heterogeneous
firms? Which workers work together, and what firm do they work for?
How do worker earnings depend on their employer, their coworkers,
as well as other firms and workers in the economy? How does the firm
value vary with the workers they hire? To answer these questions,
we study the sorting of multiple heterogeneous workers into teams
at heterogeneous firms when workers' skills are imperfect substitutes.

The traditional approach to one worker, one firm assignment problems
is to provide conditions under which sorting is assortative. \citet{Becker:1973}
shows that optimal sorting is positive under supermodular technologies
and negative under submodular technologies.\footnote{\citet{Sattinger:1993}, \citet{Chiappori:2016}, \citet{Chade:2017},
and \citet{Eeckhout:2018r} provide a comprehensive review of this
literature.} Positive sorting, or comonotonicity, readily extends to settings
with multiple workers and is generally optimal with supermodular technologies. On the other hand, negative sorting, or countermonotonicity, does
not have a simple multi-type equivalent.

An important open question is how to characterize sorting of heterogeneous
firms with multiple heterogeneous workers when technology is submodular.
We show that this problem can be analyzed as a multimarginal optimal
transport problem and fully characterize the solution.

We are the first to fully solve a sorting problem with firm heterogeneity
and multiple heterogeneous workers whose skills are imperfect substitutes,
that is, when the technology is submodular. In doing so,
this paper makes two contributions. First, we show that an equilibrium
assignment is characterized by two types of assignment regions. In
the first region, mediocre firms sort with mediocre workers and coworkers
so that output losses are equal across all these teams (mixed assignment).
In the second region, high skill workers sort with low skill coworkers
and a high productivity firm, while high productivity firms employ
low skill workers and a single high skill coworker (countermonotonic
assignment). We call this a mixed and countermonotonic assignment.
Second, we fully characterize the equilibrium assignment as well as
workers' wages and firm values. We illustrate our theory with a quantitative
application to earnings dispersion within and across U.S. firms.

We develop our findings using a transferable utility assignment model
with worker peer effects and firm heterogeneity. In order to illustrate
our theory, consider a simple production and team technology. The output of a firm is the product of the value of their project,
or productivity $z$, and the probability their worker team is able
to solve production problems $q$. When workers independently do not
know how to solve a problem with probability $x_{i}$, team quality
with workers $(x_{1},x_{2},\dots,x_{m})$ is $q=1-\prod\limits _{i=1}^{m}x_{i}$,
and output is $y(x_{1},x_{2},\dots,x_{m},z)=zq=\big(1-\prod\limits _{i=1}^{m}x_{i}\big)z$.
Production is supermodular in team quality and productivity, and submodular
in worker skills and productivity. Workers are imperfect substitutes,
with peer effects in output scaling with the project value.


We show there are two regions to an optimal assignment. The first
region is the mixed set. Mixing means that the output loss, $z\prod x_{i}$,
is identical for all teams. We start with the observation that mixing
is optimal whenever it is feasible. Equalizing output loss $z\prod x_{i}$
across teams is optimal as average expected losses attain their lower
bound if and only if output losses are equal across all teams following
Jensen's inequality. Mixing implies that firms with identical projects
may hire different teams of workers which have the same overall  quality.

Mixing, however, does not fully describe equilibrium
assignment as it is not feasible everywhere. We show that optimal
sorting features the maximal mixed set combined with countermonotonic
regions. To develop intuition for the optimality of countermonotonic
sets, we note that at the boundary of the mixed set, the
lowest value project is sorted with identical low-skill workers (high $x_i$).
Similarly, a project which has zero value employs workers who do not
know how to solve a problem ($x_i=1$). Applying this reasoning to other low
value projects outside the mixed set, these projects are countermonotonically
sorted with low-skill workers.

In order to establish equilibrium properties of worker wages and firm
values we characterize the dual problem. We show that the derivative of
the wage schedule is equal to the marginal worker product, or $-z\prod\limits _{j\neq i}x_{j}$.
The increase in output obtained by replacing a worker with a slightly
more skilled worker, keeping their coworkers and their firm unchanged,
has to equal the increase in wages necessary to hire the higher-skill
worker. More skilled workers work with coworkers and a firm that has
greater expected output losses, implying a convex wage schedule.

Our main result is derived in a setting with heterogeneous distributions
of workers and firms and teams of $n\ge2$ members. We prove that a
mixed and countermonotonic assignment is optimal, derive its dual,
and thus characterize an equilibrium. We establish the optimality of
a mixed and countermonotonic assignment by using a majorization inequality
\citep{Hardy:1929,Karamata:1932,Pecaric:1984}.\footnote{We prove there exists a mixed set by connecting our assignment model to
work on risk aggregation by \citet{Wang:2016} that describes when
there exists a dependence structure such that the sum of idiosyncratic
risks is constant.} We combine our characterization of optimal assignments
with a duality argument for multimarginal optimal transport problems \citep{Kellerer:1984}
to derive worker wages and firm values. Finally, we discuss how
our results apply to more general production structures.

Having characterized equilibrium, we quantitatively illustrate the
theory using administrative U.S. earnings records. Since a characteristic
feature of our paper among assignment models is that we obtain distinct
distributions of earnings within heterogeneous firms, we apply our
model to evaluate the dispersion in earnings within and across U.S.
firms.

Our model can generate the dispersion of earnings as well as its decomposition
between and within firms as empirically documented by \citet{Song:2019}.
We use the cross-sectional distribution of earnings and the decomposition
of earnings variation into within-firm and between-firm variation
to estimate the underlying distributions of worker skills and firm
values.  We show that only changing the distribution of firm projects
between 1981 and 2013, the share of within-firm earnings dispersion
would have increased by 22 percentage points. Changing the distribution
of workers, we show that the share of within-firm earnings dispersion
would have decreased by 38 percentage points. That is, our counterfactual
analysis shows that both the changes in the worker and project distributions
are important in the model in generating the observed change in earnings
dispersion.


\vspace{0.35cm}
\noindent \textbf{Related Literature}. There is an extensive literature on
one-to-one assignment models in the tradition of \citet{Becker:1973}.
This literature, in particular, analyzes conditions under which sorting
is positive or negative.\footnote{\citet{Boerma:2023} provides a closed-form solution to a one-to-one
assignment that is neither positive sorting nor negative sorting,
which they call composite sorting.} Our contribution is to characterize sorting for an assignment problem
with heterogeneous firms and multiple heterogeneous workers whose
skills are imperfect substitutes. We provide a complete equilibrium
characterization for this problem, and show this equilibrium features
neither positive nor negative sorting.

\noindent 




The closest to our work is a literature that studies negative sorting
with more than two agents. \citet{Ahlin:2017}, \citet{Chade:2018}
and \citet{Eeckhout:2018r} argue it is neither evident how to define
negative sorting in this case, nor how to characterize an optimal
assignment when production is submodular, and also derive insightful
partial characterizations of the equilibrium. \citet{Ahlin:2017}
also develops a quantitative application to group formation
in villages.\footnote{See also \citet{Ahlin:2015} for a study of a role of group sizes
in group lending and \citet{SaintPaul:2001} for an assignment model
with different types of intra-firm spillovers.} In our paper, we fully characterize optimal sorting for dimensions
greater than two with a submodular technology.


An alternative approach to multiworker firms in assignment models
is due to \citet{Eeckhout:2018}.\footnote{\citet{Kelso:1982} provide the gross substitutes conditions for equilibrium
existence in a related many worker-to-one firm assignment model with
finite workers and firms and a more general production technology.} They develop a model in which firms sort with a single type of worker,
and choose how many workers of this type to employ. \citet{Eeckhout:2018}
develop conditions under which sorting is positive or negative. Importantly,
they assume the firm technology is additively separable between workers
of different types, so the marginal worker product is independent
from other types of workers within the firm. This paper differs from
\citet{Eeckhout:2018} by focusing on firms of fixed size and by incorporating
peer effects. Workers are instead imperfect substitutes under a submodular
technology, and their marginal product does depend on their
coworkers.\footnote{Recent quantitative work by \citet{Jarosch:2021} and \citet{Herkenhoff:2018}
studies production and learning of workers in teams in labor markets with and without frictions.}

This paper relates to a growing literature which uses optimal transport
theory to solve economic problem, see \citet{Galichon:2018} for a
comprehensive overview. Specifically, there is an important line of
work on multidimensional sorting and multi-market problems, working
with optimal transport theory, such as  \citet{Dupuy:2014}, \citet{McCann:2015},
\citet{Lindenlaub:2017}, \citet{Chiappori:2017}, \citet{Chiappori:2017b},
\citet{Ocampo:2018}, \citet{Lindenlaub:2020} and \citet{Galichon:2021b}.
Our assignment problem instead results in a sorting problem where
each dimension of heterogeneity sorts with an endogenous multidimensional
distribution.


Another literature, following \citet{Garicano:2006}, solves hierarchical
assignment models with heterogeneous workers.\footnote{An exposition of the hierarchical assignment problem is presented
in \citet{Garicano:2004}. See \citet{Garicano:2015} for a review.} Similar to \citet{Garicano:2000}, workers differ in their knowledge,
which governs the probability that they know how to solve a production
problem. Knowledge is assumed to be cumulative, so that more skilled
workers know how to solve a problem whenever less skilled workers
do. A key implication is that production is supermodular in worker
skill, so that equilibrium sorting is positive. In line with these
papers, our workers differ in the probability that they know how to
solve production problems. Knowledge, however, is not cumulative,
allowing for the possibility that less skilled workers know how to
solve a problem when more skilled workers do not. As a result, our
production technology is not supermodular, and hence optimal
sorting is not positive.

To characterize our equilibrium, we build on the optimal transport
literature in mathematics that develops tools to solve multimarginal transport
problems.\footnote{See, for example, a review in \citet{Pass:2011}. He concludes
that the extension of the classical Monge-Kantorovich problem to the
case of more than two marginal distributions is not well understood.}
While this is an active area of recent research with a number of applications,
understanding of the problem is far from complete and is fractured
(for examples of different settings and costs, see \citet{Gangbo:1998},
\citet{Carlier:2010},  \citet{Kim:2014}, and \citet{Gladkov:2020}).
Our paper studies a submodular cost function and introduces firm production,
non-uniform distributions of workers and firms, and uses this to understand
equilibrium sorting, worker wages, and firm values. In the literature
on risk aggregation and insurance, \citet{Bernard:2014}, \citet{Embrechts:2014},
\citet{Puccetti:2015} and \citet{Wang:2016} study properties of
aggregate risk as the sum of individual risks. Similar to them, we
study the minimization of output losses (risks) incurred by teams
of workers and firms in the economy. Different from them, our interest
lie in the distribution of output losses and characterizing sorting
patterns of workers that generate maximum aggregate output. Moreover,
we characterize the dual problem to characterize the distribution
of wages and firm values.


\section{Model} \label{s:model}

We study an economy in which agents with heterogeneous skills choose their firm and coworkers to work with. This is an assignment problem extended to incorporate multiple workers in each firm. Our setup results in a sorting problem between multiple workers with each firm for a submodular production technology. 


\subsection{Environment}

\noindent \textbf{Agents}. There are $m \geq 2$ groups of risk-neutral workers and a single group of risk-neutral firms. Each group has mass one. 

Workers differ in skills. The skill is indexed by a single number $x_i \in X = [0,1]$, which captures the probability that the worker does not know how to solve a problem. For example, a worker with skill $x_i = 0.1$ knows how to solve a problem with 90 percent probability. High skill workers have low $x_i$. The distribution of each group of workers has cumulative distribution function $F_x(x_i)$ and corresponding continuous density function $f_x(x_i)$. 


Firms differ in productivity $z \in Z = [0,1]$, capturing the value of their projects. High value firms have a high $z$. The distribution of firms has cumulative distribution function $F_z(z)$ and continuous density function $f_z(z)$. The cumulative distribution functions are strictly increasing and continuous. The inverse distribution function for firms is so that $I_z(p_z)$ for percentile $p_z \in [0,1]$ is the unique number $z$ that satisfies $F_z(z) = p_z$. Thus, the percentile $p_z$ corresponds to the value $z$ in the firm distribution. The inverse distribution function for workers $I_{x}$ is defined analogously.

\vspace{0.35cm}
\noindent \textbf{Technology}. A team comprises of $m$ workers that encounter a problem in production. These workers solve a problem when at least one of them knows how to solve it. Since the probability that worker $i$ does not know the solution is $x_i$, independent of their coworker knowing the solution, the probability of successful production with workers $(x_1, x_2, \dots, x_m)$ is:
\begin{equation}
q = h(x_1,x_2, \dots, x_m) = 1 - \prod\limits_{i=1}^m x_{i} .  \label{e:team_quality}
\end{equation}
The team quality function $h$ is a symmetric, submodular function in worker types, so workers are substitutes.



A firm of type $z$ that employs a team of quality $q$ produces output according to:
\begin{equation}
g(q,z) = q  z . \label{e:yp_crude}
\end{equation}
The production function $g$ is supermodular in team quality $q$ and firm type $z$. Team quality and the project are complements in production. Combining the technology (\ref{e:yp_crude}) with the team quality function (\ref{e:team_quality}), firm output $y$ is: 
\begin{equation}
y (x_1,x_2,\dots,x_m,z) = \Big(1 - \prod\limits_{i=1}^m x_{i} \Big)  z. \label{yp}
\end{equation} 
Firm output is a submodular function in worker and firm type.\footnote{This technology nests one-to-one assignment models with positive and negative sorting. If the distribution for team quality $q$ is exogenous, sorting between teams and firms is positive. High quality teams work on the most valuable projects. On the other hand, absent firm heterogeneity, $z=\bar{z}$, sorting is negative in the case of two workers, or $m=2$. Low and high skill workers form teams.}  The marginal product of worker $x_i$ is:
\begin{equation}
m (x_i) = y_i (x_1,x_2,\dots,x_m,z) = - z \prod\limits_{j\neq i} x_{j}  \leq 0 .\label{mwp} 
\end{equation}  
When a worker is less skilled, that is, less likely to know how to solve a problem, output decreases. The marginal product of every worker depends only on the skills of their coworker and the firm they work with.



\vspace{0.35cm}
\noindent \textbf{Assignment}. An assignment prescribes for every worker coworkers to work with and a firm to work for. Given a distribution of workers $F_x$ and a distribution of firms $F_z$, the set of feasible assignment functions is $\Pi := \Pi(F_x,\dots,F_x,F_z)$ which is the set of probability measures $\pi$ on $X \times \dots \times X \times Z$ such that the marginal distributions of $\pi$ onto $X$ and $Z$ are equal to $F_x$ and $F_z$ respectively. Feasibility of an assignment function is equivalent to labor market clearing, that is, all workers and firms are sorted. 

 
\subsection{Assignment Problem}

We solve two problems to characterize an equilibrium.\footnote{The equilibrium definition is standard and is presented in Appendix \ref{a:eqdefn} for completeness.}  

\vspace{0.35cm}
\noindent \textbf{Primal Problem}. We first solve a primal problem to find an optimal sorting:
\begin{equation}
\max_{\pi \in \Pi} \int y (x_1,x_2,\dots,x_m,z ) \text{d} \pi . \label{pp}
\end{equation}
This problem is to choose an assignment $\pi$ to maximize production. It is equivalent to, in terms of choosing an optimal assignment, minimizing output losses. Output losses are the product of the project value $z$ and the probability of failure by a team of workers which is given by $x_1 \cdots x_m$, that is, a loss $x_1 \cdots x_m z$. We thus equivalently represent the planning problem as:
\begin{equation}
\min_{\pi \in \Pi} \; \int x_1 \cdots x_m z \text{d} \pi , \label{ppmin}
\end{equation}
which is an optimal transport problem in the tradition of \citet{Monge:1781} and \citet{Kantorovich:1942}, with multiple marginal distributions. The number of marginal distributions is given by $n = m+1$, comprising $m$ workers in a team and a project.




\vspace{0.35cm}
\noindent \textbf{Dual Problem}. To obtain equilibrium wages $w$ and firm values $v$, we solve a dual problem. The dual problem is to choose functions $w$ and $v$ that solve:
\begin{equation}
\min \; \sum^m_{i=1} \int w(x_i) \text{d} F_x + \int v(z) \text{d} F_z \label{e:pp_dual},
\end{equation}
subject to the constraint that $\sum w(x_i) + v(z) \geq y(x_1,x_2,\dots,x_m,z)$ for any $(x_1,x_2,\dots,x_m,z)$. 


\section{Intuition for the Main Result} \label{s:heuristic}

This section provides intuition for optimal sorting. We  consider
the case where all distributions are identical. 




The core idea of the optimal sorting is based on three observations.
First, Jensen's inequality bounds aggregate output losses from below:\footnote{For any convex function $\zeta$ and any random variable $X$, Jensen's
inequality states that $E[\zeta(X)]\geq\zeta(E[X])$. Here, $\zeta$
is the exponential function, and the random variable $X=\log x_{1}\cdots x_{m}z$
is the logarithmic output loss.} 
\begin{equation}
\int x_{1}\cdots x_{m}z\text{d}\pi\geq\exp\Big(\int\log x_{1}\cdots x_{m}z\text{d}\pi\Big). \label{e:intuition1}
\end{equation}
Second, the lower bound is independent of assignment $\pi$ due to
feasibility: 
\begin{equation}
\exp\Big(\int\log x_{1}\cdots x_{m}z\text{d}\pi\Big)=\exp\Big(\int\log x_{1}\text{d}F+...+\int\log x_{m}\text{d}F+\int\log z\text{d}F\Big). \label{e:intuition2}
\end{equation}
The right hand side is a constant $\mathcal{C}$ as $\int\log z\text{d}F$
and $\int\log x_{i}\text{d}F$ depend only on the marginal distributions.
Combining (\ref{e:intuition1}) and (\ref{e:intuition2}), the aggregate output loss is bounded below
by $\mathcal{C}$: 
\begin{equation}
\int x_{1}\cdots x_{m}z\text{d}\pi\geq\mathcal{C}.
\end{equation}
Third, by Jensen's inequality the minimum given by the right-hand
side is attained when the output loss is
equal across all teams and equal to $\mathcal{C}$. We refer to the
set of teams with the same output loss $\mathcal{C}$ as mixed.\footnote{\citet{Wang:2011} define assignments with equal output loss
as completely mixed. Such sortings are also studied by \citet{Gaffke:1981}, \citet{Ruschendorf:2002}, and \citet{Knott:2006}.} When an assignment is mixed, it is optimal as it minimizes aggregate
losses (\ref{ppmin}).

While mixing is optimal when it is feasible, it is not feasible everywhere.
Consider a project $z=0$ in a mixed assignment. Then the output loss
$x_{1}x_{2}\dots x_{m}z$ has to be equal to zero for all teams. Thus,
 mixed is not feasible everywhere.\footnote{Generally, there exists no mixed assignment for atomless distributions
with full support on the unit interval.}

We show that the optimal assignment is the largest possible set of mixed teams
combined with countermonotonic sets.\footnote{The teams $(x_{1},x_{2},\dots,x_{m},z)$ and $(\hat{x}_{1},\hat{x}_{2},\dots,\hat{x}_{m},\hat{z})$
are countermonotonic in $k$ if $x_{k}<\hat{x}_{k}$ implies that
$x_{i}\geq\hat{x}_{i}$ for all $i\neq k$ and $z\geq\hat{z}$, or
if $x_{k}>\hat{x}_{k}$ implies that $x_{i}\leq\hat{x}_{i}$ for all
$i\neq k$ and $z\leq\hat{z}$, or if $\hat{x}_{k}=x_{k}$. A set
of teams is countermonotonic in $k$ if all pairs of teams are countermonotonic
in $k$.} The intuition for this result is as follows. 


Jensen's inequality suggests that it is optimal to have the largest
set of mixed teams. Consider a mixed set with the workers and firms between percentiles $\underline{p}$ and $\bar{p}$ in their distribution, that is, the values $[I(\underline{p}),I(\bar{p})]$ from each distribution.  We now show that $\mathcal{C}=I(\underline{p})I^{m}(\bar{p})$
for the largest possible set of mixed teams. The maximal output loss on a low-value project
$z=I(\underline{p})$ is attained by pairing this project with the
lowest-skill workers $x_{i}=I(\bar{p})$ for all $i$. This team incurs
an output loss $I(\underline{p})I^{m}(\bar{p})$ and thus  $\mathcal{C}\leq I(\underline{p})I^{m}(\bar{p})$.
For the maximal feasible mixed set, this condition has to hold with
equality $\mathcal{C}=I(\underline{p})I^{m}(\bar{p})$. 

We next describe the reasoning behind the countermonotonic sets. First,
note that the lowest value project $z=I(\underline{p})$ in the mixed
set is paired with $m$ identical low skill workers $I(\bar{p})$
in this set to attain the loss $\mathcal{C}=I(\underline{p})I^{m}(\bar{p})$.
Second, it is naturally optimal to pair the lowest value project $z=0$
with $m$ lowest skill workers $x_{i}=1$, attaining zero loss.
Interpolating this logic between the project values $z=0$ and $z=I(\underline{p})$,
projects below $I(\underline{p})$ are sorted countermonotonically
to $m$ identical low skill workers. This generates a countermonotic
set for low value projects $z\in[0,I(\underline{p}))$ with $m$ identical
low skill workers $x_{i}\in(I(\bar{p}),1]$. By symmetry, high skill
workers $x_j\in[0,I(\underline{p}))$ are paired with $m-1$ low skill coworkers $x_{i}\in(I(\bar{p}),1]$ and a high value
project $z\in(I(\bar{p}),1]$.  Since both low value projects $z$ and high skill workers
$x_{j}$ for all $j\neq i$ are assigned to teams with low skill workers
$x_{i}$, percentiles $[0,p)$ of firms $z$ and high-skill workers
$x_{j}$ absorb measure $mp$ of low-skill coworker $x_{i}$. Percentile
$p\in[0,\underline{p})$ of the project distribution is thus sorted
with percentile $1-mp$ of the worker distributions. The least valuable
project $p=0$ employs the least productive workers at percentile
$1$; a project at percentile $\underline{p}$ is sorted with identical
workers ranked $1-m\underline{p}=\bar{p}$. This team incurs loss
$I^{m}(\bar{p})I(\underline{p})=\mathcal{C}$, identical to the loss
for all teams on the mixed set.

\begin{figure}[t!]
\begin{center}
\subfigure{\includegraphics[trim=0.0cm 0.0cm 0.0cm 0.0cm, width=0.56\textwidth,height=0.56\textheight,angle=270]{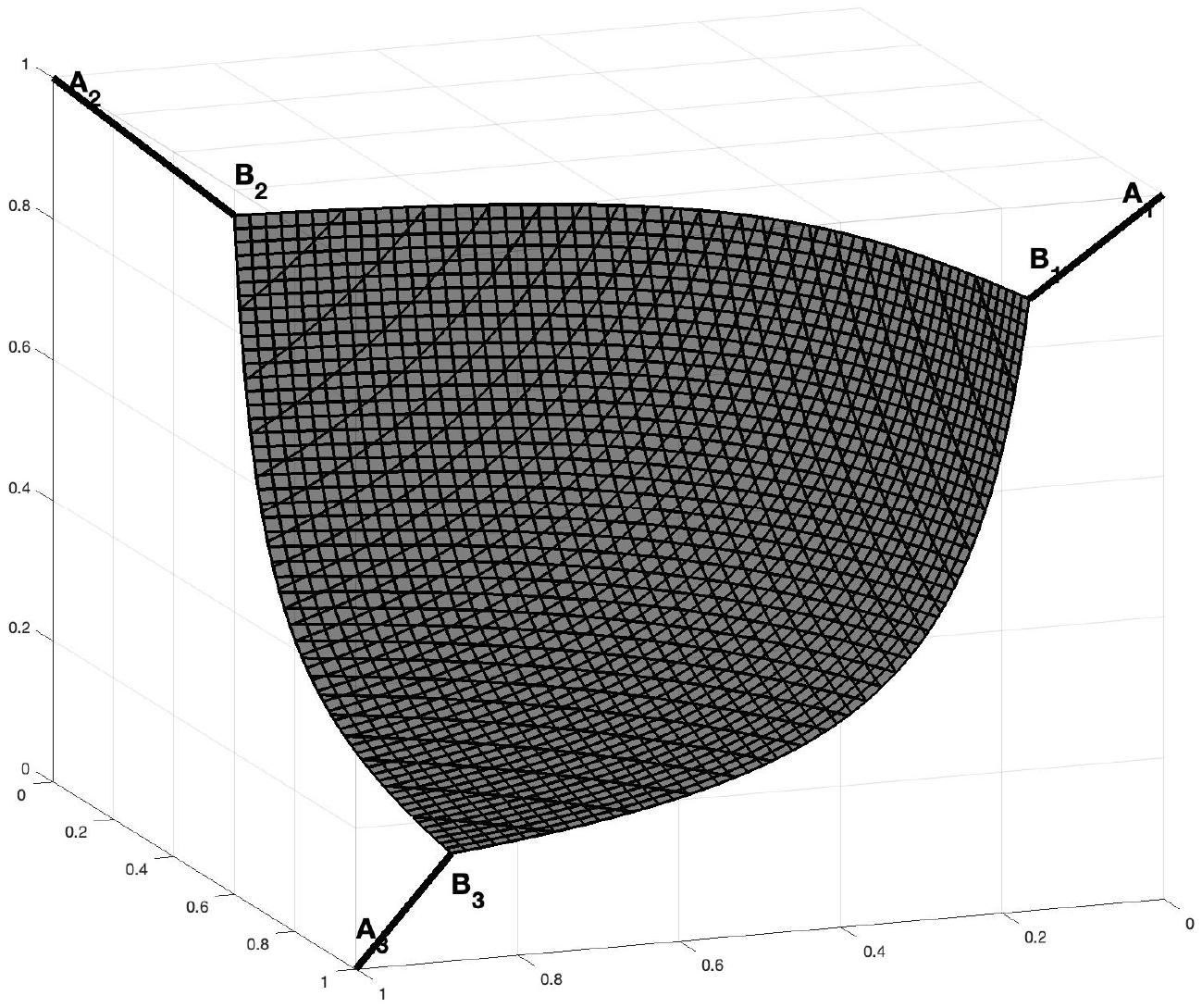}}
\end{center}\vspace{-1.5 cm}
\caption{Support of the Assignment} \label{f:support}
{\scriptsize \vspace{.2 cm} \Cref{f:support} plots an example for the support for an optimal assignment when each team employs two workers, $m=2$. The $x$-axis and $y$-axis show the skill of the worker $x_1$ and their coworker ${x_2}$, while the $z$-axis plots the value of the firm project $z$.}
\end{figure}


In sum, we provided intuition for the following solution. Pair high
skill workers with low skill workers on valuable projects, until,
at some point, mediocre workers can be paired with firms such that
output losses are identical for all mediocre teams. \Cref{f:support}
visualizes the support of this candidate solution with two workers.
The countermonotonic sets corresponding to teams with high skill workers
$x_{1}$ and $x_{2}$ is represented by $A_{1}B_{1}$ and $A_{2}B_{2}$,
while the countermonotonic set that corresponds to teams with low
value projects $z$ is $A_{3}B_{3}$. The mixed set is
represented by $B_{1}B_{2}B_{3}$.

\section{Main Results}\label{s:results}

We next show that a mixed and countermonotonic assignment solves the planning problem with heterogeneous continuous marginal distributions and characterize the dual solution for wages and firm values.

\subsection{Mixed and Countermonotonic Assignment}

We generalize the assignment in Section \ref{s:heuristic} to heterogeneous marginal distributions by developing
three insights from the intuition for the case of homogenous marginal distributions.

First, in the case of homogenous distributions, each distribution is segmented
into three parts: low values $[0,I(\underline{p}\vphantom{\bar{p}})]$,
medium values $[I(\underline{p}\vphantom{\bar{p}}),I(\bar{p})]$,
and high values $[I(\bar{p}),1]$. The measure of medium values is $\bar{p}-\underline{p} = 1-q$. With heterogeneous distributions,
we segment all distributions into three parts, each indexed by $i \in \{ 1,\dots, n\}$: low values $[0,I_{i}(\underline{p}\vphantom{\bar{p}}_{i})]$,
medium values $[I_{i}(\underline{p}\vphantom{\bar{p}}_{i}),I_{i}(\bar{p}_{i})]$,
and high values $[I_{i}(\bar{p}_{i}),1]$.\footnote{When we use $i \in \{ 1, \dots, n\}$ to enumerate the different marginal distributions, we use $1 \leq i \leq m$ to refer to the respective worker distributions, and $n$ to reflect the firm distribution.} The measure
of medium values is $\bar{p}_{i}-\underline{p}\vphantom{\bar{p}}_{i} = 1-q$, identical across distributions. 

Second, we consider the ranking of output losses across teams. In the
homogenous case, the loss is largest for teams in the mixed set which
has measure $1-q$. The remaining teams are in countermonotonic sets with measure $q$.
For each percentile $t\leq q$, there are $n$ teams located at the
$t$-th percentile of the loss distribution, one in each countermonotonic
set. Each team consists of a low value $I(\frac{t}{n})$ from one
of the distributions and $m$ high values $I(1-\frac{m}{n}t)$ from
the remaining distributions. With heterogenous distributions, the loss is as well largest for teams in the mixed set which has measure $1-q$. There
are also $n$ teams located at every percentile $t\leq q$. Each team
consists of a low value $I_{i}(\underline{t}_{i}(t))$ from one of
the distributions and $m$ high values $I_{j}(1-\bar{t}_{j}(t))$
from the remaining distributions $j\neq i$, where the functions $\underline{t}_{i}$ and $\bar{t}_{i}$ are both increasing. The function $\underline{t}_{i}(t)$
gives the percentile of the low value drawn from distribution $i$
at the $t$-th percentile of the team loss distribution. Similarly, the function
$\bar{t}_{j}(t)$ yields the percentile of the high value drawn from
distribution $j$ at the $t$-th percentile of the team loss distribution.
In the homogenous case, $\underline{t}_{i}(t)=\frac{t}{n}$ and $\bar{t}_{j}(t)=\frac{m}{n}t$.

Finally, in the case of homogeneous distributions, mixed teams contain agents in $[I(\underline{p}),I(\bar{p})]$, and
their losses are all identical and equal to $I(\underline{p})I(\bar{p})^{m}$.
With heterogenous distributions, mixed teams contain values in $[I_{i}(\underline{p}\vphantom{\bar{p}}_{i}),I(\bar{p}_{i})]$
and attain identical losses $\mathcal{C}=\prod x_{i}$, where for all $i$:
\begin{equation}
\mathcal{C} = I_{i}(\underline{p}\vphantom{\bar{p}}_{i}) \prod\limits _{j\neq i} I_{j}(\bar{p}_{j}) .\label{e:bowlloss}
\end{equation}

We use the generalization of features of the optimal sorting with homogeneous
distributions to define a mixed and countermonotonic assignment for the case
of heterogeneous distributions.

\begin{definition}\label{d:branchbowl} A \underline{mixed and countermonotonic assignment}
is a collection of continuous increasing functions $\{(\underline{t}_{i},\bar{t}_{i})\}_{i=1}^{n}$,
thresholds $\{(\underline{p}\vphantom{\bar{p}}_{i},\bar{p}_{i})\}_{i=1}^{n}$, and an
assignment $\pi\in\Pi(\{F_{i}\}_{i=1}^{n})$ satisfying:

\begin{enumerate}[noitemsep]
\item The mass of agents in the mixed sets $[I_{i}(\underline{p}\vphantom{\bar{p}}_{i}),I(\bar{p}_{i})]$ is identical for all distributions $1-q=\bar{p}_{i}-\underline{p}\vphantom{\bar{p}}_{i}$
for all $i$.
\item The support of assignment $\pi$ consists of ($1$) a mixed assignment; and ($2$)
$n$ countermonotonic sets so that for all $0\leq t\leq q$ the assignment is
countermonotonic: \\
 $b_{i}(t)=(b_{i1}(t),\dots,b_{in}(t))$ with $b_{ii}(t)=I_{i}(\underline{t}_{i}(t))$
and $b_{ij}(t)=I_{j}(1-\bar{t}_{j}(t))$ if $i\neq j$.
\item For each $i$, $b_{i}(t)$ starts at $(1,\dots,1,0,1,\dots,1)$ and
ends at $(I_{1}(\bar{p}_{1}),\dots,I_{i}(\underline{p}\vphantom{\bar{p}}_{i}),\dots,I_{n}(\bar{p}_{n}))$. Equivalently, this means that $\underline{t}_{i}(0)=\bar{t}_{i}(0)=0$,
$\underline{t}_{i}(q)=\underline{p}_{i}$, and $1-\bar{t}_{i}(q)=\bar{p}_{i}$
for each $i$.
\item For each percentile $t\leq q$, there are $n$ teams located at the $t$-th percentile, when teams are ranked according to their output losses. Team $i$ is located on the countermonotonic set at $b_{i}(t)$. All other teams are mixed and have identically large losses equal to $\mathcal{C}$ given by (\ref{e:bowlloss}).
\end{enumerate}
\end{definition}

\subsection{Planning Problem} \label{s:planningproblem}

In this section, we prove our main result that a mixed and countermonotonic assignment, following Definition \ref{d:branchbowl}, is optimal.

\vspace{0.1 cm}
\begin{theorem} \label{p:support2}
Any mixed and countermonotonic assignment $\pi \in \Pi$ solves the planning problem (\ref{pp}).
\end{theorem}

\vspace{0.1 cm}
\noindent We provide a sketch of the proof. Let $\gamma$ be the distribution of logarithmic output losses induced by assignment $\pi$. The associated inverse cumulative distribution function is $I_\gamma$. That is, logarithmic losses by a team at percentile $p$ in the loss ranking are given by $I_\gamma(p)$ when teams are sorted in increasing loss order. We show that an assignment $\pi^*$ is optimal by showing that output losses are greater under any other assignment $\pi$. That is, $\int \ell(x_1, x_2, \dots, x_n, z) \text{d} \pi^\ast \leq \int \ell(x_1, x_2, \dots, x_m, z) \text{d} \pi$ or, equivalently, $\int_0^1 \exp( I_{\gamma^\ast}(p))\text{d}p \leq \int_0^1 \exp( I_{\gamma}(p))\text{d}p$.\footnote{It is equivalent to integrate output losses with respect to the assignment function and with respect to ranks in the logarithmic loss distribution, $\int \ell(x_1,x_2,\dots,x_m,z) \text{d} \pi = \int_0^1 \exp(I_{\gamma}(p) )\text{d}p$.}

Specifically, we establish optimality by proving that logarithmic output losses under the optimal assignment are smaller in convex order than the logarithmic output losses under any other feasible assignment.\footnote{A random variable $A$ is less in convex order than a random variable $B$ if and only if for any convex function $h$: $E[h(A)] \leq E[h(B)]$. This is equivalent to $\int_0^1 h(I_A(p))\text{d}p \leq \int_0^1 h(I_B(p))\text{d}p$ given that $I$ is the quantile function.} The proof is constructed using the integral version of the majorization inequality \citep{Hardy:1929,Karamata:1932,Pecaric:1984}. 

\vspace{0.35 cm}
\noindent \textit{Majorization Inequality}. Let $h$ be a continuous convex function, and let $s_1$ and $s_2$ be non-decreasing functions on the unit interval. Suppose
\begin{align}
1. \int_0^1s_1(p) \text{d} p & = \int_0^1 s_2(p) \text{d}p \label{e:karamata1} \\ 
2. \int_0^t \hspace{0.03 cm} s_1(p) \text{d} p & \geq \int_0^t \hspace{0.03 cm} s_2(p) \text{d} p , \label{e:karamata2}
\end{align}
for all $t \in [0, 1]$. Then $\int_0^1 h(s_1(p))\text{d}p \leq \int_0^1 h(s_2(p))\text{d}p$.
\vspace{0.35 cm}

\noindent We use the majorization inequality with $h = \exp$, $s_1(p) = I_{\gamma^*}(p)$ and $s_2(p) = I_{\gamma}(p)$. The exponential function is convex, and logarithmic losses $I_\gamma$ are increasing as the logarithm of a positive increasing function is itself increasing.

For application of the majorization inequality, it suffices to verify that (\ref{e:karamata1}) and (\ref{e:karamata2}) are satisfied. The first condition (\ref{e:karamata1}) requires that aggregate logarithmic losses under the assignments are equal. As aggregate logarithmic losses are invariant across feasible assignments this condition is verified. That is, for any feasible assignment $\pi$: $\int_0^1\log I_{\gamma}(p) \text{d} p = \int \log \ell(x_1,x_2,\dots,x_m,z) \text{d} \pi = \sum_{i=1}^m \int \log x_i \text{d} F_i + \int \log z \text{d} F_z$.

The second condition (\ref{e:karamata2}) is that cumulative logarithmic losses at any percentile in the team rank $t$ are larger under optimal assignment $\pi^*$. We observe that cumulative logarithmic losses are convex in team rank $t$ as they integrate the increasing negative function $\log I_{\gamma}$. For the mixed set, that is teams ranked $t \in [q,1]$, cumulative logarithmic losses decrease linearly as losses are constant. Since cumulative logarithmic losses for the lowest ranked teams $t \in [0,q)$ are maximized by the countermonotonic assignments (\Cref{lemma:S_gamma} in Appendix \ref{proof:support2}), and since aggregate logarithmic losses are equal, cumulative logarithmic losses are larger for any team rank $t \in [0,1]$. We formalize the verification of condition (\ref{e:karamata2}) in \Cref{proof:support2}.\footnote{\Cref{p:support2} does not imply uniqueness of an optimal assignment, only that any mixed and countermonotonic assignment solves the planning problem. We provide a further characterization of the mixed and countermonotonic assignment in \Cref{pf:branchbowl}, and prove when its exists in \Cref{a:existencemc}. In Section \ref{s:dualp}, we characterize equilibrium wages and firm values, which are uniquely determined.} 

\vspace{0.35 cm}
\noindent \textbf{General Production Functions}. We next discuss how the mixed and countermonotonic sorting is also an optimal assignment with more general production technologies.


Consider a general technology $y(x_{1},\dots,x_{m},z)$. Its second-order Taylor expansion is a second-order polynomial in worker skills $\{ x_{1},\dots,x_{m}\}$ and project value $z$. Assume that the matrix of the coefficients in this expansion has rank one and is elementwise non-positive to ensure submodularity. In \Cref{a:appprod}, we establish that a mixed and countermonotonic assignment is optimal. This shows that our results generalize in the second order sense to general production functions.

In the case of teams comprising of two workers and a project, we derive a complementary generalization. Let a technology function be given by $y(x_{1},x_{2},z,x_{1}x_{2},x_{2}z,x_{1}z,x_{1}x_{2}z)$ which additionally incorporates interactions effects between workers and coworkers $x_{1}x_{2}$, and workers and project $x_{1}z$ and $x_{2}z$. Suppose its first-order Taylor expansion around the function's arguments has negative coefficients on each interaction term $(x_{1}x_{2},x_{2}z,x_{1}z,x_{1}x_{2}z)$ to ensure submodularity. In \Cref{a:appprod}, through a change of variables, we establish optimality of a mixed and countermonotonic assignment. This shows that our results extend in the first order sense to these functions of single, pairwise, and triple interactions.

\subsection{Dual Problem} \label{s:dualp}

In this section, we develop the analysis of the dual problem to characterize wages and firm values.

\vspace{0.3cm}
\noindent \textbf{Marginal Worker Product}. The marginal worker product (\ref{mwp})
is the marginal output loss induced by a worker, which is the negative
product of the project value and the coworker skills. Given an optimal
assignment, the marginal worker product can be fully described.

First, consider a high-skill worker $x_{i}\in[0,I_{i}(\bar{p_{i}})]$.
Given team rank $t$, the high-skill worker is $x_{i}=I_{i}(\underline{t}_{i}(t))$,
and their team is in a countermonotonic set. This worker is assigned
to low-skill team members, given by $I_{j}(1-\bar{t}_{j}(t))$
on a valuable project $I_{n}(1-\bar{t}_{n}(t))$. The marginal product
for high-skill workers is $m_{i}(x_{i})=-\prod\limits _{j\neq i}I_{j}(1-\bar{t}_{j}(t))$.

Second, consider mediocre workers $x_{i}\in[I_{i}(\underline{p}\vphantom{p}_{i}),I_{i}(\bar{p}_{i})]$.
Mediocre workers work in mixed teams with identical output losses
$\ell(x_{1},\dots,x_{m},z)=x_{1}\dots x_{m}z=\mathcal{C}$
across all such teams. Given this constant output loss, the marginal
product for mediocre workers is $m_{i}(x_{i})=-\mathcal{C}/x_{i}$.

Finally, consider low-skill workers $x_{i}\in[I_{i}(\bar{p}_{i}),1]$.
The rationale for their marginal product is a direct consequence of
the assignment of high-skill workers described above. There are exactly
$m$ teams containing this low-skill worker with the same loss. Let $k$
be the index of the high-skill worker in one such team, so that $I_{k}(\underline{t}_{k}(t))$
is the high-skill worker. The value of team member $j$  is given
by $I_{j}(1-\bar{t}_{j}(t))$ for all $j\neq k$. Thus, the marginal
product of the worker is: $-I_{k}(\underline{t}_{k}(t))\prod\limits _{j\ne i,k}I_{j}(1-\bar{t}_{j}(t))$.
The description of the marginal product is summarized by Proposition
\ref{prop:marginal}.


\vspace{0.15cm}
\begin{proposition}{\textit{Marginal Product}}. \label{prop:marginal}
The marginal product of worker $x_{i} = I_i(p_i)$ is: 
\begin{align}
m_{i}(x_{i})=\begin{cases}
\hspace{0.121cm}-\;\prod\limits _{j\neq i}I_{j}(1-\bar{t}_{j}(t))\hspace{4.31cm}\text{if }x_{i}\in[0,I_{i}(\bar{p}_{i})]\\
\vspace{0.15cm}\hspace{0.12cm}-\;\mathcal{C}\big/x_{i}\hspace{6.29cm}\text{if }x_{i}\in[I_{i}(\underline{p}_{i}),I_{i}(\bar{p}_{i})]\\
\vspace{0.10cm}\hspace{0.12cm}-\;I_{k}(\underline{t}_{k}(t))\prod\limits _{j\neq i,k}I_{j}(1-\bar{t}_{j}(t))\hspace{2.54cm}\text{if }x_{i}\in[I_{i}(\bar{p}_{i}),1]\label{e:m}
\end{cases}
\end{align}
where the constant $\mathcal{C}$ is given by (\ref{e:bowlloss}).
\end{proposition}

Proposition \ref{prop:marginal_product} establishes properties of
the marginal worker product.

\begin{proposition}{\textit{Properties of the Marginal Product}}.
\label{prop:marginal_product} The marginal product $\{m_{i}\}$ is
continuous and increasing. \end{proposition}

\noindent We now give an outline of the proof to this proposition. First, continuity
within all three segments directly follows as the distribution functions
$\{F_{i}\}$, their inverse functions $\{I_{i}\}$, and $\{\underline{t}_{i},\bar{t}_{i}\}$
are all continuous. In addition, at the point $x_{i}=I_{i}(\underline{p}\vphantom{p}_{i})$
the values of the high-skill and mediocre workers align; at the point
$x_{i}=I_{i}(\bar{p}_{i})$ the values of the low-skill workers and
the mediocre workers align. Hence, the marginal product $\{m_{i}\}$
is continuous across all skill levels.\footnote{The marginal worker product $m_i$ is identical for all worker distributions, or $m_i = m$ as we show in \Cref{a:existencemc}.}


Second, the most insightful part of this proposition is that the marginal
product is increasing. This fact follows from the necessary condition
that total losses do not decrease by exchanging workers between
their teams. For any two teams $(x_{1},\dots,x_{m},z)$ and
$(\hat{x}_{1},\dots,\hat{x}_{m},\hat{z})$ in the optimal assignment $\ell(x_{1},\dots,x_{m},z)+\ell(\hat{x}_{1},\dots,\hat{x}_{m},\hat{z})\leq\ell(x_{1},\dots,\hat{x}_{i},\dots,x_{m},z)+\ell(\hat{x}_{1},\dots,x_{i},\dots,\hat{x}_{m},\hat{z})$.
Given the loss function $\ell(x_{1},\dots,x_{m},z)=x_{1}\dots x_{m}z$ and the marginal
worker product (\ref{mwp}): 
\begin{align}
\big(x_{i}-\hat{x}_{i}\big)\big(z\prod_{j\neq i}x_{j}-\hat{z}\prod_{j\neq i}\hat{x}_{j}\big) & \leq0\hspace{1.2cm}\Longleftrightarrow\hspace{1.2cm}(x_{i}-\hat{x}_{i})(m_{i}(\hat{x}_{i})-m_{i}(x_{i}))\leq 0 . \label{e:stability_explicit}
\end{align}
When a worker is more skilled, $\hat{x}_{i}<x_{i}$, (\ref{e:stability_explicit})
implies that they work with coworkers on a project that has a greater
output loss, $\hat{z}\prod\limits _{j\neq i}\hat{x}_{j}\geq z\prod\limits _{j\neq i}x_{j}$.
Equivalently, when a worker is more skilled, their marginal product
(\ref{mwp}) is more negative, and the marginal product is thus increasing.

\vspace{0.35cm}
\noindent \textbf{Wages}. Given the marginal product we formulate equilibrium
wages. To simplify notation, define the surplus $S$ as output
minus payments to workers and firms: 
\begin{equation}
S(x_{1},\dots,x_{m},z)=y(x_{1},\dots,x_{m},z)-\sum_{i=1}^{m}w(x_{i})-v(z).\label{e:surplus_equation}
\end{equation}
The constraint to the dual problem is that the surplus $S(x_{1},\dots,x_{m},z)$
is negative for any team $(x_{1},\dots,x_{m},z)\in X\times\dots\times X\times Z$.
In Appendix \ref{s:duality_kellerer} we show that the dual solution
exists in the class of continuous functions, hence the surplus function
(\ref{e:surplus_equation}) is continuous also.

Using the duality result of \citet{Kellerer:1984}, the optimal value
to the planner problem equals the optimal value for the dual problem,
$\int y(x_{1},\dots,x_{m},z)\text{d}\pi=\sum\int w(x_{i})\text{d}F_{x}+\int v(z)\text{d}F_{z}$,
where $\pi$ solves the planning problem and $w$ and $v$ solve the
dual problem. Given feasibility and the definition of the surplus,
$\int S(x_{1},\dots,x_{m},z)\text{d}\pi=0$. Since the surplus is
non-positive for every team by the constraint to the dual problem
(\ref{e:pp_dual}), the surplus is zero almost everywhere
with respect to the assignment $\pi$.

We next characterize the derivative of the wage schedule and the
firm value function.

\vspace{0.1cm}
 \begin{proposition}{\textit{Marginal Wages and Firm Values}}.\label{prop:cont_derivative}
Wages and firm values are both continuously differentiable, with derivatives
$w'(x)=m(x)$ and $v'(z)=m_z(z)+1$. \end{proposition} \vspace{0.1cm}

\noindent The proof is presented in Appendix \ref{proof:cont_derivative}.
Workers' marginal earnings equal their marginal product. The marginal
increase in output must equal the marginal increase in the cost necessary
to recruit a higher skill worker. While the marginal product depends only on peer effects, the equilibrium level of these peer
effects is directly related to the worker's own skill through Proposition
\ref{prop:marginal}.

While a worker's marginal earnings reflect their marginal product,
earnings levels reflect the marginal product of all workers that are
more skilled. Specifically, workers incur an earnings penalty relative
to the most skilled worker. Wages are characterized in Proposition
\ref{prop:wages}.


\vspace{0.1cm}
\begin{proposition}{\textit{Wages}}.\label{prop:wages} The wage
schedule is, up to additive worker constant $\mathcal{C}_{w}$, given
by: 
\begin{equation}
w(x)=\mathcal{C}_{w}+\int_{0}^{x}m(s)\text{d}s.\label{e:wage_equation}
\end{equation}
\end{proposition}

\vspace{0.1cm}
\noindent To obtain intuition, observe the surplus is zero almost everywhere
with respect to the assignment $\pi$, $S(x_{1},\dots,x_{m},z)=z\big(1-\prod x_{i}\big)-\sum w(x_{i})-v(z)=0$.
Moreover, since $w$ is a differentiable function, fix the coworkers
$x_{\neg i}$ and the project $z$. Since the surplus function
is negative $S(x_{1},\dots,x_{m},z)\leq0$, the equilibrium team $(x_{1},\dots,x_{m},z)$
is a local maximum with respect to the surplus function, and hence:
\begin{equation}
S_{i}(x_{1},\dots,x_{m},z)=0=-z\prod\limits _{j\neq i}x_{j}-w'(x_{i})\hspace{0.7cm}\implies\hspace{0.7cm}w'(x_{i})=m(x_{i}),
\end{equation}
where the final equality follows from the definition of the marginal
worker product (\ref{mwp}). Since there are no excess resources and
the surplus is differentiable, every worker's marginal earnings is
their marginal product. By integrating, equilibrium wages follow (\ref{e:wage_equation}).

\noindent 

The marginal cost of workers is equal to their marginal product $w'(x)=m(x)$
for all $x\in(0,1)$. By Proposition \ref{prop:marginal_product},
the negative marginal worker product is continuous and increasing,
establishing that equilibrium wages are decreasing and convex.


\vspace{0.3cm}
\noindent \textbf{Firm}. We next characterize firm values. The intuition for
the firm value is similar to the intuition for wages. If there are
no excess resources almost everywhere for assignment $\pi$, then,
by fixing a team of workers as well as by differentiability
of the firm value $v$: 
\begin{equation}
S_{n}(x_{1},\dots,x_{m},z)=0=1-\prod\limits _{j}x_{j}-v'(z)\hspace{0.7cm}\implies\hspace{0.7cm}v'(z)=m_{z}(z)+1,
\end{equation}
where the final equality follows from the definition of the marginal
product (\ref{mwp}). A firm's marginal reward is its marginal product.
When a firm undertakes a marginally more valuable project, its output
is its expected probability of success, or team quality, $1-\prod\limits _{j}x_{j}$.
In equilibrium, the firm's marginal product is similar to the worker
product, up to the unit constant. Firm values are characterized in
Proposition \ref{c:firm_value}.

\begin{proposition}{\textit{Firm Value}}. \label{c:firm_value}
The firm value $v$ is, up to additive firm constant $\mathcal{C}_{v}$,
given by: 
\begin{equation}
v(z)=\mathcal{C}_{v}+\int_{0}^{z}k(s)\text{d}s,\label{e:firm_value}
\end{equation}
where $k(s)=m_{z}(s)+1$ is the marginal firm product. \end{proposition}
\vspace{0.1cm}

\noindent By Proposition \ref{c:firm_value}, the firm value shares
properties with the wage function. By Proposition \ref{prop:marginal_product},
the marginal firm product $k$ is continuous and increasing. The marginal
worker product $m$ is continuous and increasing, and $m(z)\geq-1$
implying $v'(z)\geq0$. The firm value function is thus increasing
and convex in project value $z$.

Given the description of wages in Proposition \ref{prop:wages} and
the firm value function in Proposition \ref{c:firm_value}, we observe
that for any constants $\mathcal{C}_{w}$ and $\mathcal{C}_{v}$ that
satisfy: 
\begin{equation}
0=\mathcal{C}_{w}+\mathcal{C}_{v}\big/m+\int_{0}^{1}m_{x}(s)\text{d}s\label{e:constants}
\end{equation}
the functions $(w,v)$ are a dual solution. Equation (\ref{e:constants})
ensures the surplus is zero for the team $(x_{1},\dots,x_{m},z)=(1,\dots,1,0)$.

\vspace{0.35 cm}
\noindent \textbf{Summary}. In \Cref{s:planningproblem} we show that the mixed and countermonotonic assignment
solves the assignment problem. In \Cref{s:dualp} we construct the solution to the corresponding dual problem. Since the surplus is zero everywhere in the support of an optimal assignment, the value for the planning problem and the dual
problem coincide. In \Cref{a:eqdefn} we show that any mixed and countermonotonic assignment $\pi$
 together with dual functions $(w,v)$ satisfying (\ref{e:constants})
are therefore an equilibrium.\footnote{Equation (\ref{e:constants}) determines the equilibrium in terms
of wages and firm values up to a constant. When functions $w$ and
$v$ solve the dual problem, so do functions $\hat{w}$ and $\hat{v}$
which differ from $w$ and $v$ only in terms of the constants as
long as they satisfy (\ref{e:constants}).}


\section{Quantitative Analysis}

We evaluate the ability of the model to quantitatively generate the observed wage dispersion, as well as its decomposition into within and between-firm components. 


\subsection{Data} \label{s:data}


We use data on worker earnings for 1981 and 2013 from the US Social Security Administration of \citet{Song:2019}. The data considers employed individuals between 20 and 60 years of age.\footnote{Individual earnings are derived from W-2 forms, which capture compensation for labor services as defined by the Internal Revenue Service. An individual is considered employed when their earnings exceed the minimum wage earned full-time for one quarter, that is, for 520 hours. While the baseline sample of \citet{Song:2019} considers firms with over 20 employees, we confirm our results extend to firms of all sizes. Workers with multiple jobs in a year are linked to the firm that provides their largest share of labor earnings. All amounts are in 2013 dollars.} The main reason for using these records is that the data matches the universe of individuals to their employer which allows us to analyze sorting patterns. Importantly, the matched employer-employee dataset provides information on the earnings distribution within every firm. 

\begin{figure}[t!] 
\begin{center} 
\subfigure{\includegraphics[trim=0.0cm 0.0cm 0.0cm 0.0cm, width=0.48\textwidth,height=0.27\textheight]{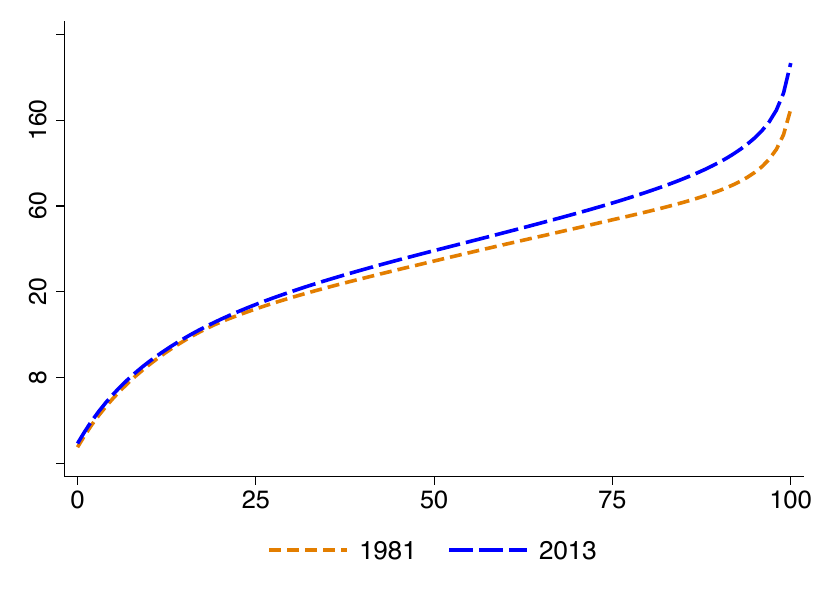}} 
\end{center}\vspace{-1 cm}
\caption{Distribution of Individual Earnings} \label{f:annual_earnings}
{\scriptsize \vspace{.2 cm} Figure \ref{f:annual_earnings} displays the distribution of annual earnings in the SSA, replicating Figure I in \citet{Song:2019}. The figure shows individual earnings levels, on the $y$-axis in thousands of dollars, for each percentile of the income distribution, on the $x$-axis.}
\end{figure}

Figure \ref{f:annual_earnings} shows individual earnings by percentile of the earnings distribution for 1981 and 2013. In 1981, earnings vary between 18 thousand at the 25th percentile, to 31 thousand at the median, to 89 thousand at the 95th percentile. By 2013, earnings are 19 thousand at the 25th percentile, 35 thousand at the median, and 132 thousand at the 95th percentile. Figure \ref{f:annual_earnings} shows that the earnings distribution features considerable dispersion in levels, which has increased over time. Between 1981 and 2013, earnings in the bottom third of the earnings distribution have seen little change, while earnings at the top strongly increased. 

We decompose the overall variance of log earnings into within- and between-firm components. Let $w_{ij}$ be log earnings of worker $i$ at firm $j$. Earnings can be written as $w_{ij} = \bar{w}_{j} + ( w_{ij}  - \bar{w}_{j} )$, where $\bar{w}_j$ is average log earnings within firm $j$. The total variance is then:
\begin{equation}
\text{Var}(w_{ij}) = \text{Var}(\bar{w}_{j}) + \sum_j \theta_j \text{Var}(w_{ij} | i \in j) \label{e:var_decomposition} \hspace{0.07 cm} ,
\end{equation}
where $\theta_j$ is the employment share of firm $j$. The first term is the variance of mean earnings across firms, or the between-firm variance. The second term is the average of within-firm dispersion of employee earnings weighted by employment. 

The variance of log earnings in 1981, 0.65 log points, is for one-third attributed to differences in mean earnings between firms and for two-thirds to within-firm earnings dispersion as documented by \citet{Song:2019}. From 1981 to 2013 the variance of log earnings increased by 0.20 log points to 0.85. Two-thirds of this increase is attributed to increased differences in mean earnings across firms, while a third is attributed to an increase in within-firm earnings dispersion. 

\subsection{Model}

We assess the ability of our model to generate dispersion in earnings, its decomposition between and within firms, as well as changes in the distribution of earnings over time. We then use the calibrated model to structurally decompose increased earnings dispersion between changes on the worker side and changes on the firm side of the labor market. 

In our model, the worker skill distribution $F_x$ and the project value distribution $F_z$ are both exogenous. We parameterize the distribution for workers as $\text{Beta}(\alpha_x,\beta_x)$ and the distribution for projects as $\text{Beta}(\alpha_z,\beta_z)$. We consider teams with two workers.


\begin{figure}[t!]
\begin{center}
\subfigure[{\footnotesize 1981}]{
\includegraphics[trim=0.0cm 0.0cm 0.0cm 0.0cm, width=0.48\textwidth,height=0.27\textheight]{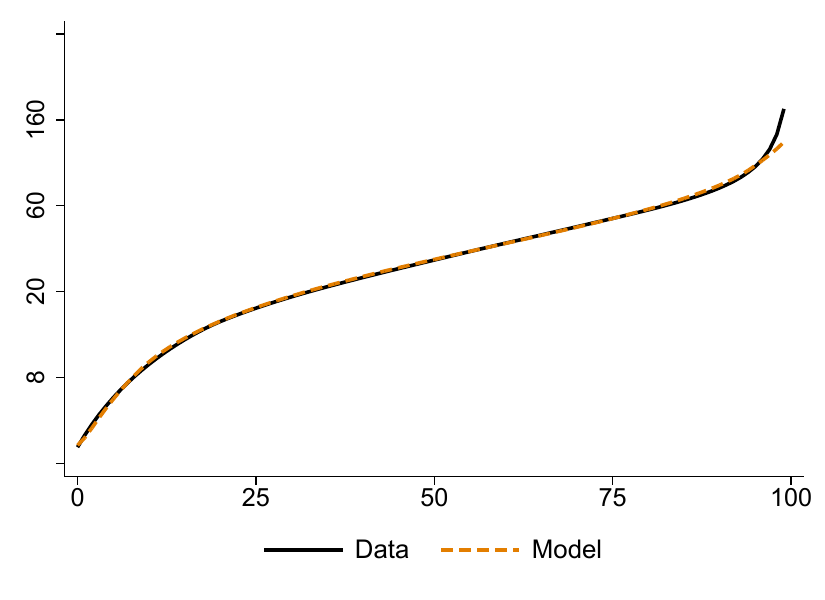}}
\subfigure[{\footnotesize 2013}]{
\includegraphics[trim=0.0cm 0.0cm 0.0cm 0.0cm, width=0.48\textwidth,height=0.27\textheight]{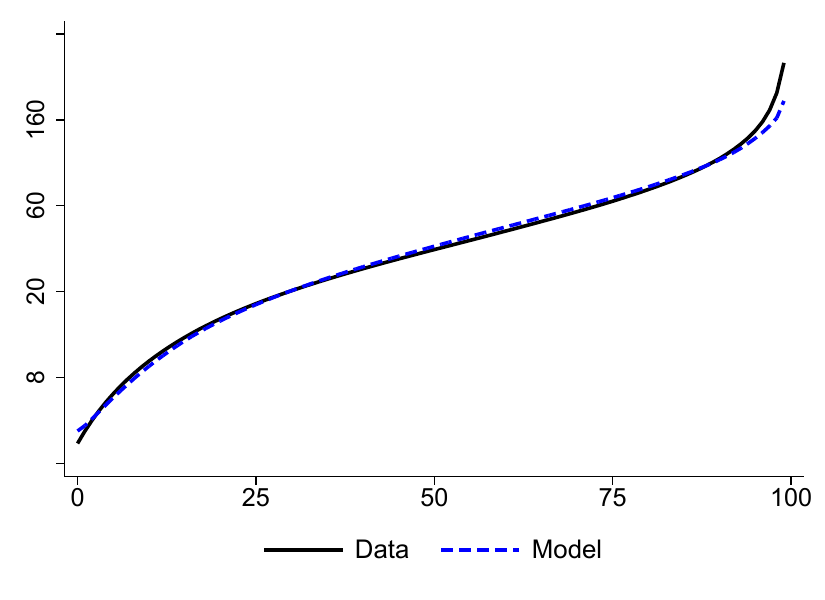}}
\end{center}
\vspace{-0.5 cm}
\caption{Data and Model Distribution of Individual Earnings} \label{f:annual_earnings_dm}
{\scriptsize \vspace{.2 cm} \Cref{f:annual_earnings_dm} compares the empirical earnings distribution to the model earnings distribution. The empirical distributions, which we display by solid black lines, follow \Cref{f:annual_earnings}, while the model distributions are presented by dashed lines. The left panel shows the empirical and model distribution for 1981, the right panel for 2013.}
\end{figure}

We use the cross-sectional distribution of earnings and the decomposition of the earnings variation into within-firm and between-firm variation to inform the worker and project distribution. We estimate the underlying distributions using earnings in the cross-section using the data underlying Figure \ref{f:annual_earnings}. The resulting model fit for cross-sectional earnings is summarized by Figure \ref{f:annual_earnings_dm}. Figure \ref{f:annual_earnings_dm} shows that our framework generates the observed dispersion in earnings, as well as changes in the distribution of earnings over time while somewhat underestimating earnings at the very top percentiles in both 1981 and 2013.\footnote{The parameters which govern the Beta distribution for worker skills are (0.96, 2.85) in 1981 and (1.89, 1.48) in 2013. The estimated parameters for the project value distribution are (0.98, 2.83) in 1981 and (1.66, 1.63) in 2013.} 

\begin{table}[!t]
\def\arraystretch{1.65}
\begin{center}
\caption{Model and Data Earnings Decomposition} \label{t:decomposition}
\begin{tabular}{l|ccc|ccc}
\hline  \hline
\hspace{2.5 cm} & \multicolumn{3}{c|}{Data}     & \multicolumn{3}{c}{Model}     \\
Moment  \hspace{0.38 cm} & \hspace{0.38 cm} 1981 \hspace{0.38 cm}  & \hspace{0.38 cm} 2013 \hspace{0.38 cm} & \hspace{0.38 cm} change \hspace{0.38 cm} & \hspace{0.38 cm} 1981 \hspace{0.38 cm} & \hspace{0.38 cm} 2013 \hspace{0.38 cm} & \hspace{0.38 cm} change \hspace{0.38 cm} \\ \hline
Between & 0.34     & 0.42 & \hspace{0.04 cm} 0.08 & 0.34     & 0.43 & \hspace{0.04 cm} 0.09 \\
Within & 0.66     & 0.58 & -0.08 & 0.66     & 0.57 & -0.09 \\ 
 \hline \hline
\end{tabular}
\end{center}
{\scriptsize \vspace{.1 cm} \Cref{t:decomposition} compares the empirical and model decomposition of earnings dispersion. The left panel shows the empirical decomposition of earnings variation in dispersion between firms and dispersion within firms following (\ref{e:var_decomposition}), while the right panel shows the model analog.}
\end{table}

\Cref{t:decomposition} displays the decomposition for the variance of log earnings between 1981 and 2013 in the data and the model. Our stylized model can account for the earnings decomposition in 1981. The fifth column shows that the model correctly attributes a third of the earnings variation to dispersion between firms, and two thirds to dispersion within firms. The model generates an increase in the share of the between-firm variance from 0.34 to 0.43 percentage points between 1981 and 2013, while capturing the overall increase in earnings dispersion as shown in \Cref{f:annual_earnings_dm}.\footnote{\citet{Gavilan:2012} studies the change in between and within firm earnings dispersion in response to the change in the price of capital and in response to the changes in the skill distribution in the setup of \citet{Kremer:1996} with endogenous capital.}


\begin{table}[!t]
\def\arraystretch{1.65}
\begin{center}
\caption{Model and Data Coworker Earnings} \label{t:coworker_salary}
\begin{tabular}{l|ccc|ccc}
\hline  \hline
\hspace{1.5 cm} & \multicolumn{3}{c|}{Data}     & \multicolumn{3}{c}{Model}     \\
Percentile  \hspace{0.38 cm} & \hspace{0.38 cm} 1981 \hspace{0.38 cm}  & \hspace{0.38 cm} 2013 \hspace{0.38 cm} & \hspace{0.38 cm} change \hspace{0.38 cm} & \hspace{0.38 cm} 1981 \hspace{0.38 cm} & \hspace{0.38 cm} 2013 \hspace{0.38 cm} & \hspace{0.38 cm} change \hspace{0.38 cm} \\ \hline
25 & 10.01     & 10.07 & 0.06 & 10.19    & 10.30 & 0.11 \\
50 & 10.30     & 10.45 & 0.15 & 10.44     & 10.63 & 0.19 \\
75 & 10.56     & 10.75 & 0.19 & 10.56     & 10.83 & 0.27 \\ 
90 & 10.67     & 10.98 & 0.31 & 10.46     & 10.70 & 0.24 \\ 
 \hline \hline
\end{tabular}
\end{center}
{\scriptsize \vspace{.1 cm} \Cref{t:coworker_salary} compares the model and data in terms of mean coworker earnings. The left panel shows mean coworker earnings at selected percentiles of the individual earnings distribution. The right panel shows their model analog. All numbers are log points.}
\end{table}

In addition to studying the variance decomposition of log earnings, we evaluate non-targeted average log earnings of coworkers at different percentiles of the earnings distribution. The results are in \Cref{t:coworker_salary}. Model and data align qualitatively up to the 75th percentile with quantitative deviations of 0.20 log points. At the top of the skill distribution, the data display a stronger positive correlation between worker earnings compared to our more negatively sorted model economy.

\subsection{Counterfactual}

The rise in earnings dispersion between 1981 and 2013 could be driven by changes in the worker or the firm side of the labor market, or a combination thereof. We use our framework to evaluate the drivers of the increase in earnings dispersion. To do this, we analyze counterfactual changes in earnings dispersion by only changing the worker distribution and by only changing the project value distribution.

\begin{table}[!t]
\def\arraystretch{1.65}%
\begin{center}
\caption{Model and Data Earnings Decomposition} \label{t:decomposition_structural}
\begin{tabular}{l|ccc|cc|cc}
\hline  \hline
\hspace{1.5 cm} & \multicolumn{3}{c|}{Model}     & \multicolumn{2}{c|}{Firm Effect} & \multicolumn{2}{c}{Worker Effect} \\
Moment  \hspace{0.22 cm} & \hspace{0.22 cm} 1981 \hspace{0.22 cm}  & \hspace{0.22 cm} 2013 \hspace{0.22 cm} & \hspace{0.22 cm} change \hspace{0.22 cm} & \hspace{0.22 cm} 2013 \hspace{0.22 cm} & \hspace{0.22 cm} change \hspace{0.22 cm} & \hspace{0.22 cm} 2013 \hspace{0.22 cm} & \hspace{0.22 cm} change \hspace{0.22 cm} \\ \hline
Between & 0.34     & 0.42 & \hspace{0.04 cm} 0.08  & 0.12 & -0.22  & 0.72 & 0.38  \\
Within & 0.66    & 0.58 &  -0.08      & 0.88 & \hspace{0.04 cm} 0.22  & 0.28 & -0.38  \hspace{-0.02 cm} \\ 
 \hline \hline
\end{tabular}
\end{center}
{\scriptsize \vspace{.1 cm} \Cref{t:decomposition_structural} compares the baseline and counterfactual model decomposition of earnings dispersion. The left panel shows the baseline model decomposition of earnings as in \Cref{t:decomposition}, while the middle and right panel show counterfactual decompositions. For the firm effect counterfactual, we evaluate the model with the worker distribution for 1981 and the firm distribution for 2013, while the worker effect counterfactual evaluates the model using the firm distribution for 1981 and the worker distribution for 2013.}
\end{table}

\Cref{t:decomposition_structural} shows a structural decomposition of changes in earnings dispersion from 1981 to 2013. The left panel repeats the baseline decomposition, while the middle panel and right panel present counterfactual results. The firm effect counterfactual evaluates the model for the distribution of workers in 1981 and the firm distribution in 2013. The middle panel demonstrates that by only changing the distribution of firm projects, the share of within-firm earnings dispersion would have increased by 22 percentage points. The worker effect counterfactual similarly evaluates the model using the project distribution of 1981 and the worker distribution of 2013. The right panel shows that the share of the within-firm earnings dispersion would have decreased by 38 percentage points. That is, our counterfactual analysis shows that both the changes in the worker and project distributions between 1981 and 2013 are important in the model in generating the observed change in earnings dispersion.\footnote{The quantitative analysis thus far assumes that teams are of size three – two workers and a project. In Appendix \ref{a:appquant}, we demonstrate that our model can also account for the decomposition of earnings within and across firms with teams of sizes four, five, and six. We also find that the non-targeted moments of average coworker earnings are robust to team size. As the number of team members increases, the firm and worker effect become more muted.}

\section{Conclusion}

We provide a complete solution to an assignment problem with heterogeneous firms and multiple heterogeneous workers whose skills are imperfect substitutes, that is, for a submodular technology. 


\clearpage


\vspace{-0.0cm}

{ {\
\bibliographystyle{econometrica}	
\vspace{-0.0cm}
\baselineskip20.0pt
\bibliography{BTZ_bib}
} }

\clearpage

\pagebreak

\renewcommand{\theequation}{A.\arabic{equation}} \setcounter{equation}{0}
\renewcommand{\thefigure}{A.\arabic{figure}}\setcounter{figure}{0}
\renewcommand{\thetable}{A.\arabic{table}}\setcounter{table}{0}
\setcounter{page}{1}
\newpage \appendix
\newpage

\begin{center}
{\Large Sorting with Teams\\}
\bigskip
{\Large Appendix \\}
\bigskip
{\large  Job Boerma, Aleh Tsyvinski and Alexander Zimin \\}
\bigskip
{\large November 2023}
\end{center}
\vspace{0.5cm}



\section{Proofs}

In this appendix, we formally prove the results in the main text.

\subsection{Equilibrium, Planning Problem and Duality} \label{a:eqdefn} 

We formally define an equilibrium for the economy.

\vspace{0.35 cm}
\noindent A firm with project $z$ chooses workers $x_1, x_2, \dots, x_m$ to maximize profits taking the wage schedule for workers $w$ as given. The firm problem is:
\begin{align}
v (z) = \max_{\{ x_i \}} \; y (x_1,\dots,x_m,z) - \sum^m_{i=1} w (x_i)  . \label{e:firm_problem}
\end{align}
Worker $x_i$ chooses to work for firm $z$ with coworkers $\{ x_j \}$ for $j \neq i$ to maximize their wage income. The worker takes the wage schedule for their coworkers and the firm value $v$ as given:
\begin{align}
w (x_i) = \max_{\{ x_j \} , z} \;  y (x_1,x_2,\dots,x_m,z ) - \sum^m_{j \neq i} w (x_j) -  v (z)  . \label{e:worker_problem}
\end{align}

An equilibrium is a wage function $w$, firm value function $v$, and feasible assignment function $\pi$, such that firms solve their profit maximization problem (\ref{e:firm_problem}), workers solve the worker problem (\ref{e:worker_problem}), and satisfy a feasibility constraint:
\begin{equation}
\int y(x_1,x_2,\dots,x_m,z ) \text{d} \pi = \sum_{i=1}^m \int w(x_i) \text{d} F_x + \int v(z) \text{d} F_z,
\end{equation}
which states that total output produced, $\int y(x_1,x_2,\dots,x_m,z ) \text{d} \pi$, equals total output distributed to workers and firms.

\vspace{0.5 cm}
\noindent We use the following relation between the planning problem and the dual formulation.

\vspace{0.05cm}
\begin{lemma}\label{l:duality}
Let $\pi \in \Pi (F_x, \dots, F_x, F_z)$ be a joint probability measure, and $(w,v)$ be functions such that the surplus $S(x_1,\dots,x_m,z) =y(x_{1},\dots,x_{m},z)-\sum w(x_{i})-v(z)\leq 0$ for all $(x_1,\dots,x_m,z)$. If there is a set $M \subset X \times \dots \times X \times Z$ such that $S(x_1,\dots,x_m,z) = 0$ on $M$ with the additional property that $\pi(M) =1$, then assignment $\pi$ is a primal solution and functions $w$ and $v$ are a dual solution.
\end{lemma}


\vspace{0.05cm}
\noindent Since the assignment in \Cref{l:duality} is concentrated on set $M$, we refer to $M$ as the set of potential matches. The proof to \Cref{l:duality} only makes use of a notion of weak duality. 

\vspace{0.5 cm}
\noindent \textbf{Weak Duality}. Let $\pi \in \Pi (F_{x_1}, \dots, F_{x_m}, F_z)$ be a joint probability measure, and $f_i \in L^1(X, F_{x_i})$, $g \in L^1(Z, F_{z})$ be integrable functions such that $y(x_1,\dots,x_m,z) \leq \sum\limits_{i=1}^m f_i(x_i) + g(z)$ for all teams $(x_1,\dots,x_m,z)$. Then\footnote{The minimum and maximum are attained by strong duality. We discuss this notion of duality in \Cref{s:duality_kellerer}.}
\begin{equation}
\min_{\{f_i \},g} \; \sum_{i=1}^m \int f_i(x_i) \text{d}F_{x_i} + \int g(z) \text{d}F_{z} \; \geq \; \max_{\pi \in \Pi} \; \int y(x_1,\dots,z) \text{d}\pi .  \label{e:weak_duality}
\end{equation}

\vspace{-0.15 cm}
\begin{proof}
For any functions $\{f_i\}$ and $g$ so that $y(x_1,\dots,x_m,z) \leq \sum\limits_{i=1}^m f_i(x_i) + g(z)$ we have:
\begin{equation*}
\max_{\pi \in \Pi} \int y(x_1,\dots,x_m,z) \text{d}\pi \leq \int \Big( \sum\limits_{i=1}^m f_i(x_i) + g(z) \Big)\text{d}\pi =  \sum_{i=1}^m \int f_i(x_i) \text{d}F_{x_i} + \int g(z) \text{d}F_{z} ,
\end{equation*}
where the equality follows as $\pi \in \Pi (F_{x_1}, \dots, F_{x_m}, F_z)$. Since the inequality holds for any $(\{f_i\},g)$ and $g$ it holds for $(\{f_i\},g)$ that minimize the right-hand side.
\end{proof}

\vspace{0.05 cm}
\noindent We use weak duality to establish \Cref{l:duality} by contradiction. 

\vspace{0.35 cm}
\noindent \textbf{Proof of \Cref{l:duality}}. Suppose by contradiction that $\hat{\pi}$ does not solve the primal problem. Then there exists another probability measure $\pi$ such that 
\begin{align*}
\max_{\pi \in \Pi} \int y(x_1,\dots,x_m,z) \text{d}\pi & > \int y(x_1,\dots,x_m,z) \text{d}\hat{\pi} = \sum_{i=1}^m \int \hat{f}_i(x_i) \text{d}F_{x_i} + \int \hat{g}(z) \text{d}F_{z}\notag \\
& \hspace{3.29 cm} \geq \min_{\{f_i \},g} \; \sum_{i=1}^m \int f_i(x_i) \text{d}F_{x_i} + \int g(z) \text{d}F_{z} ,
\end{align*}
where the equality follows by the assumption. This contradicts weak duality (\ref{e:weak_duality}).

Suppose by contradiction that the functions $(\{\hat{f}_i\},\hat{g})$ do not solve the dual problem. Then there exists functions $(\{f_i\},g)$ such that
\begin{align*}
\min_{\{f_i \},g} \; \sum_{i=1}^m \int f_i(x_i) \text{d}F_{x_i} + \int g(z) \text{d}F_{z} & < \sum_{i=1}^m \int \hat{f}_i(x_i) \text{d}F_{x_i} + \int \hat{g}(z) \text{d}F_{z} \notag \\ 
& =  \int y(x_1,\dots,x_m,z) \text{d}\hat{\pi} \leq \max_{\pi \in \Pi} \int y(x_1,\dots,x_m,z) \text{d}\pi,
 \end{align*}
where the equality follows by the assumption. This inequality contradicts weak duality (\ref{e:weak_duality}).


\vspace{0.35 cm}
\noindent While \Cref{l:duality} connects the planning problem and the dual formulation, it does not connect either to our definition of an equilibrium. The link is Proposition \ref{prop:equilibrium}.\footnote{Similar to \citet{Gretsky:1992}, we decentralize the solutions to the primal and dual problem as a competitive equilibrium. A similar argument for one-to-one assignment problems is in Proposition 2.3 of \citet{Galichon:2018} and discussed in \citet{Chade:2017}.}

\vspace{0.05cm}
\begin{proposition} \label{prop:equilibrium}
Let $\pi$ solve the planning problem and let $(w,v)$ solve the dual problem. Then, wage schedule $w$, firm value function $v$, and assignment $\pi$ are an equilibrium.
\end{proposition}

\vspace{-0.15cm}
\begin{proof}
The proof follows from the constraints on the dual problem, or $w (x_i) \geq y(x_1,\dots,x_m,z) - \sum\limits_{j \neq i}^m w(x_j) - v (z)$, which imply: \vspace{-0.2 cm}
\begin{equation}
w (x_i) \geq \hspace{0.03 cm} \max_{x_{\neg i},z} \Big(  y(x_1,\dots,x_m,z) - \sum\limits_{j \neq i}^m w(x_j) - v (z) \Big).
\end{equation} 
The equality follows since $y(x_1,\dots,x_m,z) = \sum\limits_{i=1}^m w(x_i) + v(z)$ on the set of potential matches $M$ on which the assignment function is concentrated. The argument for the firm value function $v$ is identical. The equality also implies that the goods market clears.\end{proof}

\noindent Proposition \ref{prop:equilibrium} implies that assignment function $\pi$, wage schedule $w$, and firm value function $v$ are an equilibrium.\footnote{We observe that an economy with $m$ identical worker distributions $F_x$ with mass equal to one for each position is equivalent to an economy with a single distribution of workers $F_x$ with mass equal to $m$. Owing to the symmetry of worker skills in the team quality function (\ref{e:team_quality}), these are equivalent. Intuitively, any equilibrium worker assignment for a planning problem with distinct worker distributions can be made symmetric. By assigning $1/m$ of the mass of equilibrium worker pairings to one distribution and half to another, we obtain $m$ identical worker distributions. The optimum value is unaffected as the same distribution of worker pairings can be made due to symmetry.} Given this connection, we call $w$ the wage schedule and $v$ the firm value.\footnote{The proposition proposes a clear path for characterizing equilibrium. We separately solve the planning problem and the dual problem, which are linked through the primitive project value distribution, worker distribution, and technology. Moreover, Proposition \ref{prop:equilibrium} shows that the decentralized equilibrium coincides with the solution to the planning problem so the competitive equilibrium is efficient.}

\subsection{Majorization Inequality} \label{proof:support2}

We verify the second assumption to the majorization inequality (\ref{e:karamata2}). The goal is to show that the mixed an countermonotonic assignment is optimal. We do not require all properties of the mixed and countermonotonic assignment described in Definition \ref{d:branchbowl}. Specifically, we prove Proposition \ref{p:support}.

\vspace{0.1 cm}
\begin{proposition} \label{p:support}
Suppose assignment $\pi \in \Pi$ satisfies two conditions. First, the loss of the $1-q$ largest teams are the same. Second, consider the first $t \leq q$ teams. Worker $x_i \in F_i$ belongs to one of them if and only if $x_i \in [0,I_i(\underline{t}_i(t))]$ or $x_i \in [I_i(1 - \bar{t}_i(t)),1]$. Then, $\pi \in \Pi$ solves the planning problem (\ref{pp}).
\end{proposition}
\vspace{0.1 cm}


\noindent For any assignment $\pi$, let $S_\gamma(t)$ denote the cumulative logarithmic losses for the $t$ teams with the smallest losses:
\begin{equation}
S_\gamma(t) = \int_0^t \log I_{\gamma}(p) \text{d}p.
\end{equation}
To verify (\ref{e:karamata2}) we compare $S_{\gamma^*}$ to $S_\gamma$. 

Suppose there exists a mixed and countermonotonic assignment $\pi^*$ with an associated loss distribution $\gamma^*$. Next, we prove $S_{\gamma^*}$ is the maximal cumulative logarithmic loss in line with the second assumption to the majorization inequality, that is, for any assignment $\pi$, $S_{\gamma^*}(t) \geq S_{\gamma}(t)$ 
for all $t \in [0, 1]$. To establish this, we first bound the cumulative logarithmic losses in the countermonotonic sets of $\pi^*$ from below.

\begin{lemma}\label{lemma:S_gamma}
For any feasible assignment $\pi$ and for all $t \in [0, q]$, $S_{\gamma}(t) \leq S_{\gamma^*}(t)$.
\end{lemma}

\begin{proof}
The function $S_{\gamma}$ is the sum of logarithmic losses for the $t$ teams with smallest output losses. For any other set $U$ with mass $\pi(U) = t$, by definition:
\begin{equation}
S_{\gamma}(t) \leq \int \log \ell(x_1, \dots, x_m,z) \text{d} \pi_U . \label{a:defn1}
\end{equation} 
where $\pi_U$ restricts assignment $\pi$ to the set $U$.

Define the set of teams with the $\underline{t}_k(t)$ highest skill workers in dimension $k$ as: 
\begin{equation}
V_k := \{(x_1, \dots, x_m,z) \in [0, 1]^n \; \vert \; 0 \leq x_k \leq I_k(\underline{t}_k(t))\} 
\end{equation}
The measure of the set $V_k$ with respect to assignment $\pi$ is $\underline{t}_k(t)$. The union of sets $V = V_1 \cup \dots \cup V_n$ thus has measure of at most $t$ with respect to assignment $\pi$. We construct a set $U \supset V$ such that $\pi(U) = t$ by adding some additional set of teams to the set $U$.

Given that $\pi_U$ represents the restriction of the assignment function $\pi$ to the specified set $U$, we use (\ref{a:defn1}) to write:
\begin{equation*}
S_{\gamma}(t) \leq \int \log \ell(x_1 , \dots, x_m,z) \text{d} \pi_U = \sum_{k = 1}^m \int \log x_k \text{dPr}_{x_k}(\pi_U) + \int \log z \text{dPr}_{z}(\pi_U),     
\end{equation*}
where $\text{Pr}_{x_k}$ is the projection of measure $\pi_U$ onto the $k$-th dimension. To characterize the upper bound, we bound $\int \log x_k \text{dPr}_{x_k}(\pi_U)$. Since $V_k \subset U$, the density function that corresponds to measure $\text{Pr}_{x_k}(\pi_U)$ is equal to the skill density $f_k(x)$ for all $0 \leq x \leq I (\underline{t}_k(t))$ and is less than skill density $f_k$ otherwise. Hence, 
\begin{align*}
\int \log x_k \text{dPr}_{x_k}(\pi_U) & \leq \int_0^{I_k(\underline{t}_k(t))} \log (x_k) f_k(x_k)\,\text{d}x_k + \int_{I_k(1 - \bar{t}_k(t))}^1 \log (x_k) f_k(x_k)\,\text{d}x_k\\
&= \int_0^{\underline{t}_k(t)} \log I_k(p) \,\text{d}p + \int_{1 - \bar{t}_k(t)}^1 \log I_k(p) \,\text{d}p.
\end{align*}
where the inequality follows as the second term on the right-hand side captures the contribution to logarithmic losses by the $\bar{t}_k(t)$ lowest skill workers. The equality follows by change of variables. Summing over all $k$ dimensions, we obtain:
\begin{equation}
S_{\gamma}(t) \leq \sum_{k = 1}^n \Big(\int_0^{\underline{t}_k(t)} \log I_k(p) \,dp + \int_{1 - \bar{t}_k(t)}^1 \log I_k(p) \,dp\Big) = S_{\gamma^*}(t),   
\end{equation}
which was what we wanted.
\end{proof}

\begin{lemma} \label{lemma:Sbounds}
For any feasible assignments $\gamma$ and for all $t \in [q, 1]$,
\[
S_{\gamma}(t) \le S_{\gamma^*}(t).
\]
\end{lemma}
\begin{proof}
We first observe that cumulative logarithmic loss $S_{\gamma}$ is a convex function as it is the integral over an increasing function. Since $S_{\gamma}$ is convex, for all $t \in [q, 1]$: 
\begin{equation}
S_\gamma(t) \leq \frac{1 - t}{1 - q}S_\gamma(q) + \frac{t - q}{1 - q}S_\gamma(1) \leq \frac{1 - t}{1 - q}S_{\gamma^*}(q) + \frac{t - q}{1 - q}S_{\gamma^*}(1) \label{e:Shat1}
\end{equation}
where the final inequality follows because $S_\gamma(q) \leq S_{\gamma^*}(q)$ by \Cref{lemma:S_gamma}, and $S_{\gamma}(1) = S_{\gamma^*}(1)$ by the verification of first assumption to the majorization inequality.

The integral over logarithmic losses $S_{\gamma^*}(t)$ is linear for teams with the highest losses $t \in [q, 1]$ by the condition that the loss of the $1-q$ largest teams are the same. Hence, for teams $t \in [q, 1]$:
\begin{equation}
S_{\gamma^*}(t) = \frac{1 - t}{1 - q}S_{\gamma^*}(q) + \frac{t - q}{1 - q}S_{\gamma^*}(1) \label{e:Shat2}
\end{equation}
Combining (\ref{e:Shat1}) and (\ref{e:Shat2}) the result follows.
\end{proof}

\vspace{0.15 cm}
\noindent We combine \Cref{lemma:S_gamma} and \Cref{lemma:Sbounds} to verify the second assumption of the majorization inequality, which concludes the proof.

\subsection{Characterization of Mixed and Countermonotonic Assignment} \label{pf:branchbowl}

The mixed-and-countermonotonic assignment in Definition \ref {d:branchbowl} may appear general, but in fact, the countermonotonic sets $b_i(t)$ fully determine the assignment $\pi$.

First, at each percentile $t \le q$ in the distribution of output losses, there are $n$ teams: one from each countermonotonic set. The logarithmic output loss at percentile $t \leq q$, which we denote $l(t)$, equals\footnote{For each percentile $t \leq q$, consider all $n$ teams at this percentile. The logarithmic loss of the $i$-th team equals $l_i(t) = \log I_i(\underline{t}_i(t)) - \log I_i(1 - \bar{t}_i(t)) + \sum \log I_j(1 - \bar{t}_j(t))$. Since the $n$ teams are all at the $t$-th percentile in the distribution of output losses their losses are identical, and hence $\lambda(t)$ is independent of $i$. Using the skill gap $\lambda$, we express $l(t)$ as in equation (\ref{eq:logloss_definition}). Summing the definition of the skill gap over $i \leq m$: $m \lambda (t) = \sum_{i \leq m} ( \log I_i(1 - \bar{t}_i(t)) - \log I_i(\underline{t}_i(t)) )$, and by adding and subtracting $\log I_n(\underline{t}_n(t))$: $m \lambda (t) = \log I_n(\underline{t}_n(t)) + \sum_{i \leq m} \log I_i(1 - \bar{t}_i(t)) - \sum \log I_i(\underline{t}_i(t))$ and thus equation (\ref{eq:logloss_definition}).\label{e:footnotelosses}}
\begin{equation}\label{eq:logloss_definition}
l(t) = m \lambda(t) + \sum_{i = 1}^{n}\log I_i(\underline{t}_i(t)).
\end{equation}
where $\lambda(t)$ is the skill gap between the high-skill and low-skill worker that are in teams at percentile $t \leq q$, which is identical across distributions $i$:
\begin{equation}\label{eq:lambda_definition}
\lambda(t) = \log I_i(1 - \bar{t}_i(t)) - \log I_i(\underline{t}_i(t)) .
\end{equation}
Since the functions $I_i$, $\underline{t}_i$, and $\bar{t}_i$ are continuous, the functions $\lambda$ and $l$ are continuous too. It follows from the definition of $l$ that the logarithmic loss function is increasing. 

Next, we connect the functions $\underline{t}_i$ and $\bar{t}_i$. Teams in the countermonotonic sets have low output losses and pair a single low value from one distribution with high values from the remaining distributions. Since the quality of the high-skill worker decreases with team loss percentile $t$, that is, $\underline{t}_i(t)$ increases from $\underline{t}_i(0) = 0$, it follows that the $t$ teams with the lowest losses draw a single high-skill worker from each distribution, so that for $t \leq q$:
\begin{equation}\label{eq:t_underbar_sum}
\underline{t}_1(t) + \dots + \underline{t}_{n}(t)= t.
\end{equation}
Equivalently, high values from the $i$-th distribution, starting from $\bar{t}_i(t)=0$, are paired into teams with low values in the remaining countermonotonic sets. Thus, for $t \leq q$:
\begin{equation}\label{eq:t_bar_equation}
\bar{t}_i(t) = \sum\limits_{j \neq i}\underline{t}_j(t) = t - \underline{t}_i(t).
\end{equation}
where the final equality follows from (\ref{eq:t_underbar_sum}).

Finally, we consider the part of the assignment $\pi$ that is located on the mixed set. Given functions $\{ \underline{t}_i \}^{n}$, we define for $x \in [0,q]$, the average logarithmic loss for mediocre workers and jobs $h(x) = \frac{1}{1 - x}\sum\int_{\underline{t}_i(x)}^{\bar{t}_i(x)} \log I_i(t) \text{d} t$.
By construction, the logarithmic loss of any team located on mixed set $B$ is the same and equal to $h(q)$. We can compute it directly:
\begin{equation}\label{eq:lh_relation}
l(q) = \frac{1}{1-q} \int_B \log x_1\dots x_{m}z \text{d} \pi = \frac{1}{1 - q}\sum_{k = 1}^{n} \int_{\underline{p}_k}^{\bar{p}_k} \log x_k\, \text{d} F_k(x_k) = h(q).
\end{equation}
This equation asserts continuity of the assignment: the output loss of teams at the end of countermonotonic sets, $l(q)$, equals the loss of teams located on the mixed set, $h(q)$. 

If a mixed and countermonotonic assignment exists, then equations (\ref{eq:logloss_definition}) to (\ref{eq:lh_relation}) are satisfied. Together with the ability to mix mediocre workers and projects, (\ref{eq:logloss_definition}) to (\ref{eq:lh_relation}) are in fact sufficient to show that a mixed and countermonotonic assignment exists as we show in Proposition \ref{e:propbbapp}.\footnote{Formally, we say that the distributions $\{ F_i \}^{n}_{i=1}$ restricted to the mediocre intervals $[I_i(\underline{t}_i(q)), I_i(1 - \bar{t}_i(q))]$ can be mixed if there exists an assignment between these distribution such that the output loss of each team is identical.}

\begin{proposition}\label{e:propbbapp}
Consider continuous non-negative increasing functions  $\{ \underline{t}_i \}^{n}_{i=1}$ and a mass $0 \leq q \leq 1$ which is located on the countermonotonic sets. There exists a mixed and countermonotonic assignment defined by these functions if and only if
\begin{enumerate}[noitemsep]
    \item For each $t \in [0, q]$, $\underline{t}_1(t) + \dots + \underline{t}_{n}(t) = t$ (\ref{eq:t_underbar_sum});
    \item For each $t \in [0, q]$, $\bar{t}_i(t) = t - \underline{t}_i(t)$ are non-negative and increasing (\ref{eq:t_bar_equation});
    \item For each $t \in [0, q]$, the skill gap between the high-skill and low-skill worker $\lambda(t) = \log I_i(1 - \bar{t}_i(t)) - \log I_i(\underline{t}_i(t))$ is independent of $i$ (\ref{eq:lambda_definition});
    \item For each $t \in [0, q]$, the function \(
    l(t) = m\lambda(t) + \sum\limits_{i = 1}^{n}\log I_i(\underline{t}_i(t))
    \)
    is increasing in $t$ (\ref{eq:logloss_definition});
    \item $h(q) = l(q)$ (\ref{eq:lh_relation});
    \item The restriction of $\{ F_i \}^{n}_{i=1}$ to the intervals $[I_i(\underline{t}_i(q)), I_i(1 - \bar{t}_i(q))]$ is mixable.
\end{enumerate}
\end{proposition}

\begin{proof} To prove Proposition \ref{e:propbbapp} it remains to be shown that if Conditions 1 to 6 apply, then there exists a mixed and countermonotonic assignment defined by the functions $\{ \underline{t}_i \}$ and a mass $q$ located on countermonotonic sets. The proof is split into two parts. First, we construct the assignment $\pi$ associated with $\{ \underline{t}_i \}$ and show that it is feasible. Second, we show that assignment $\pi$ is a mixed and countermonotonic assignment. While the result is rather intuitive, this technical proof verifies this formally.


\vspace{0.5 cm}

\noindent We construct the mixed and countermonotonic assignment completely using the functions $\{ \underline{t}_i \}$ and the mass $q$ located on the countermonotonic sets. To understand that the assignment is fully determined, first observe that we can construct the increasing functions $\{ \bar{t}_i \}$ by (\ref{eq:t_bar_equation}). Given the functions $\{ \underline{t}_i, \bar{t}_i \}$ and a mass $q$, we use the definition of countermonotonic sets in Definition \ref{d:branchbowl} to construct countermonotonic sets. 


We next identify the threshold percentiles $\bar{p}_i$ and $\underline{p}\vphantom{p}_i$ and show that these threshold percentiles are inside the unit interval. We denote $\underline{p}\vphantom{p}_i = \underline{t}_i(q)$ and $\bar{p}_i = 1 - \bar{t}_i(q) = \underline{p}\vphantom{p}_i + (1-q)$, where the second equality follows from equation (\ref{eq:t_bar_equation}) evaluated at $t = q$. The threshold $\underline{p}\vphantom{p}_i$ is non-negative since $\underline{t}_i$ is a non-negative function, which also implies that $\bar{p}_i = \underline{p}\vphantom{p}_i + (1-q)$ is non-negative. Next, since $\bar{t}_i(q) \ge 0$ by (\ref{eq:t_bar_equation}) it implies that $\bar{p}_i \le 1$, and therefore $\underline{p}\vphantom{p}_i = \bar{p}_i - (1-q) \le 1$. Hence, both $\underline{p}\vphantom{p}_i$ and $\bar{p}_i$ are thresholds inside the unit interval. 

Given the threshold percentiles $\{ \underline{p}\vphantom{p}_i, \bar{p}_i \}$, we know the support of the mixed set. So far, this procedure only suggests a support of an assignment. A key step is to show there exists a measure concentrated on this support. 

The first step in showing there exists a measure concentrated on the suggested support is to reparameterize the countermonotonic set $b_i(t)$ from percentiles in the distribution of team losses $t$ to percentiles in the distribution of the $i$-th worker, which we here denoted by $q_i$. Since the function $\underline{t}_i$ is increasing, its inverse $\underline{t}^{-1}_i$ is well-defined. Function $\underline{t}_i(t)$ takes percentile $t$ of team losses and returns the percentile of the worker in the $i$-th distribution that is part of this team. The inverse function takes the percentile of the worker in the $i$-th distribution and returns the percentile in the distribution of team losses. This percentile can thus directly be substituted into the definition of a countermonotonic set $b_i$. For each $i$, we formally consider the mapping 
\begin{equation*}
\tilde{b}_i(q_i) = b_i(\underline{t}^{-1}_i(q_i)),
\end{equation*}
where $q_i \in [0, \underline{p}\vphantom{p}_i]$ and $b_i$ is the countermonotonic set defined in Definition \ref{d:branchbowl}.


Let $\lambda_i$ be the Lebesgue measure restricted to the interval $[0, \underline{p}\vphantom{p}_i]$, that is, all high-skill workers in distribution $i$ up to percentile $\underline{p}\vphantom{p}_i$. Define the corresponding countermonotonic set assignment $\pi_i$ as the push-forward image of the measure $\lambda_i$ under the mapping $q_i \to \tilde{b}_i(q_i)$. We take $\{ \lambda_i \}^n_{i=1}$ and will then  analyze the corresponding marginal distributions of these countermonotonic sets $\pi_{b} = \pi_1 + \dots + \pi_{n}$. First, we find $\tilde{b}_{ii}(q)$:
\begin{equation}\label{eq:left_marginal_proj}
\tilde{b}_{ii}(q) = b_{ii}(\underline{t}^{-1}_i(q)) = I_i \big(\underline{t}_i \big(\underline{t}^{-1}_i(q) \big)\big) = I_i(q),
\end{equation}
where the second equality follows by definition of $b_{ii}$ in Definition \ref{d:branchbowl}. In words, the $i$-th marginal distribution of assignment $\pi_i$ is the push-forward image of the uniform measure on $[0, \underline{p}\vphantom{p}_i]$ under the mapping $x \to I_i(x)$, which is the distribution $F_i$ restricted to the interval $[0, I_i(\underline{p}\vphantom{p}_i)]$.

How many low skill workers are taken by the countermonotonic sets? We assign low skill workers when we construct countermonotonic sets corresponding to the high skill workers in the remaining distributions. We next calculate how many low skill workers from distribution $k$ are assigned according to the assignment function.

For each percentile $s_k \in [\bar{p}_k, 1]$ among the low skill workers in the marginal distribution of worker $k$, we find the measure of the set $S_k = \{x \in [0, 1]^{n} \mid x_k \geq I_k(s_k)\}$ with respect to the measure $\pi_{i}$ for some $i \neq k$. In words, how many low-skilled workers from the $k$-th distribution $x_k \geq I_k(s_k)$ are taken by countermonotonic set $\pi_i$. First, since the function $\bar{t}_k$ is increasing from $0$ to $1 - \bar{p}_k$ by (\ref{eq:t_bar_equation}), its inverse $\bar{t}_k^{-1}$ is well-defined on the interval $[0, 1 - \bar{p}_k]$. We can thus write:
\begin{align*}
x_k =\tilde{b}_{ik}(q_i) = b_{ik}(\underline{t}^{-1}_i(q_i)) \geq I_k(s_k)  &\Longleftrightarrow I_k \big(1 - \bar{t}_k(\underline{t}_i^{-1}(q_i))\big) \geq I_k(s_k) \Longleftrightarrow \bar{t}_k \big(\underline{t}_i^{-1}(q_i)\big) \leq 1 - s_k\\ & \Longleftrightarrow q_i \leq \underline{t}_i\big(\bar{t}_k^{-1}(1 - s_k)\big).
\end{align*}
where $q_i$ is again the percentile of the high-shill workers from the $i$-th distribution, and the second equality follows (\ref{eq:left_marginal_proj}). Worker $x_k$ is in the set of lowest skill workers if we are looking at a percentile $q_i$ that satisfies the final inequality. This means that the measure of $S_k$ with respect to $\pi_i$ is equal to $\underline{t}_i(\bar{t}_k^{-1}(1 - s_k))$. Thus, the measure of $S_k$ with respect to $\pi_{b}$ is:
\begin{equation*}
\pi_{b}(S_k) = \sum_{i \ne k}\underline{t}_i\big(\bar{t}_k^{-1}\big(1 - s_k\big)\big) = \bar{t}_k\big(\bar{t}_k^{-1}(1 - s_k)\big) = 1 - s_k.
\end{equation*}
where the second equality follows from the definition of $\bar{t}_k$ in equation (\ref{eq:t_bar_equation}). Thus, the $k$-th marginal distribution of $\pi_{b}(S_k)$ has the same cumulative distribution as $F_k$ restricted to $[I_k(s_k), 1]$. For each $s_k$, the mass of low-skill workers above the $s_k$-th percentile that we assign is exactly $1 - s_k$. Together with \cref{eq:left_marginal_proj} this implies 
\begin{equation}\label{eq:branch_proj}
\mathrm{pr}_i(\pi_{b}) = F_i \big\vert_{[0, I_i(\underline{p}\vphantom{p}_i)] \cup [I_i(\bar{p}_i), 1]}.
\end{equation}

\noindent Now, let $\pi_{c}$ be an assignment that pairs mediocre workers in teams with the same total output loss, and denote $\pi = \pi_{b} + \pi_{c}$. \Cref{eq:branch_proj} and Condition 6 directly imply that $\pi\in \Pi(\{ F_i \}^n_{i=1})$.


\vspace{0.5 cm}
\noindent The second part of the proof verifies that the constructed assignment is indeed a mixed and countermonotonic assignment. Given the constructed assignment $\pi$, we check that all properties of the mixed and countermonotonic assignment in Definition \ref{d:branchbowl} are satisfied.

The first property is directly verified since we discussed above that $\bar{p}_i = \underline{p}\vphantom{p}_i + (1-q)$.

To verify the second property, we validate both the countermonotonic sets and the mixed set. It follows directly from their construction that the countermonotonic sets have the required form. Next, we verify the mixed set. For the verification of the mixed set there are two properties to validate. First, that the values $x_i$ are bounded between $I_i(\underline{p}\vphantom{p}_i)$ and $I_i(\bar{p}_i)$ is immediate as these are the mediocre workers and jobs. 


Second, we show the mixed teams attain a constant, in particular, $I_i(\underline{p}\vphantom{p}_i)\prod\limits_{j \ne i} I_j(\bar{p}_j)$. To validate that the loss of any team on the mixed set, $L_{B}$, equals this constant, we use the definition of the average logarithmic loss for mediocre workers and jobs $h(x) = \frac{1}{1 - x}\sum\int_{\underline{t}_i(x)}^{\bar{t}_i(x)} \log I_i(t) \text{d} t$, we have $\log L_{B} = h(q)$. By the fifth condition (\ref{eq:lh_relation}), $h(q) = l(q)$, and by the fourth and third condition, we write:
\begin{align*}
l(q) &= \log I_i(\underline{t}_i(q)) + \sum_{j \neq i}\log I_j(1 - \bar{t}_j(q)) = \log I_i(\underline{p}\vphantom{p}_i) + \sum_{j \neq i} \log I_j(\bar{p}_j).
\end{align*}
Hence we establish the loss on the mixed set attains $I_i(\underline{p}\vphantom{p}_i)\prod\limits_{j \ne i} I_j(\bar{p}_j)$, verifying the second property.

The third property of Definition \ref{d:branchbowl} need not be verified directly as it is implied by the statement of the countermonotonic sets in the second property.

Finally, for each percentile $t$, we find the mass of teams with output losses below that of the team located at point $b_i(t)$. The logarithmic loss of this team equals $l(t)$. Since the function $l(t)$ is increasing by (\ref{eq:logloss_definition}) and $\log L_{B} = l(q) > l(t)$, the logarithmic loss of any team from the mixed set is larger than the logarithmic loss of the team $b_i(t)$. 


We next find all the teams with a loss below the loss of team $b_i(t)$ located on the $i$-th countermonotonic set. The logarithm of the loss of team $b_i(x)$ is equal to $l(x)$ by Conditions 3 and 4 in Proposition \ref{e:propbbapp}. Since the loss function $l(t)$ is increasing, $l(x) \leq l(t) \Longleftrightarrow x \leq t$. As a result, the mass on the countermonotonic set $b_i$ such that $x \leq t$ is equal to $\underline{t}_i(t)$. In turn, this implies that the mass of teams with logarithmic losses below the logarithmic loss of $b_i(t)$ is $\sum\limits_{i = 1}^{n}\underline{t}_i(t) = t$, where the equality follows by the first condition in Proposition \ref{e:propbbapp}. This means that $b_i(t)$ in indeed located at the $t$-th quantile of the teams distribution.\end{proof}

\subsection{Existence of the Mixed and Countermonotonic Assignment} \label{a:existencemc}


In this appendix, we establish the existence of a mixed and countermonotonic assignment.\footnote{For multi-sided assignment games \citet{Alkan:1988} and \citet{Quint:1991} show that stable assignments may not exist. \citet{Sherstyuk:1999} analyzes multi-sided assignment games with supermodularity and provides conditions under which stable assignments exists.} We prove our results under Assumption \ref{assumption1}.\footnote{When the probability density function is differentiable, Assumption \ref{assumption1} is equivalent to bounding the elasticity of the density function from below by $-1$. All distributions with non-decreasing density functions trivially satisfy this restriction. Assumption \ref{assumption1} is satisfied by Beta distribution, truncated exponential and lognormal distributions for a subset of parameters, amongst others. Naturally, Assumption \ref{assumption1} holds for linear combinations of density functions that satisfy it.}

\vspace{0.10 cm}
\begin{assumption}\label{assumption1} 
For all $x_i$, $x_i f_i(x_i)$ is non-decreasing. \label{e:xfx}
\end{assumption}
\vspace{0.10 cm}

To construct the optimal assignment, we find continuous non-negative increasing functions $\{ \underline{t}_i\}$ and a threshold $q$ satisfying all the conditions of Proposition \ref{e:propbbapp}. Conditions 1 and 3 in Proposition \ref{e:propbbapp} can be considered a system of functional equations:
\begin{align}
t & = \underline{t}_1(t) + \dots + \underline{t}_{n}(t), \tag{\ref{eq:t_underbar_sum}} \\
\lambda(t) & = \log I_i(1 - \bar{t}_i(t)) - \log I_i(\underline{t}_i(t)) \tag{\ref{eq:lambda_definition}}
\end{align}
for all $i$. The structure of the proof is as follows. First, we show when the system of functional equations (\ref{eq:lambda_definition}) and (\ref{eq:t_underbar_sum}) has a unique solution. Second, we show this unique solution satisfies the remaining Conditions 4 to 6.\footnote{Condition 2 is implied by Condition 1 and the monotonicity of $\{ \underline{t}_i \}$. Specifically, $\bar{t}_i(t) = \sum_{j\neq i} \underline{t}_j(t)$ is increasing because $\underline{t}_j(t)$ are increasing.} That is, Conditions 4 to 6 are then properties of the solution.

To show that the system of functional equations specified by Conditions 1 and 3 has a unique solution, we first introduce new notation to transform the system of equations. 

Specifically, we introduce new functions $\tilde{y}_i(t)$, which capture the negative logarithmic skills of the high-skilled workers on the $t$-th team located on the $i$-th countermonotonic set:
\begin{equation}
\tilde{y}_i(t) = -\log I_i(\underline{t}_i(t)). \label{e:ytildedef}
\end{equation}
Since $0 \leq I_i \leq 1$, we note $0 \leq \tilde{y}_i < \infty$. We also introduce the distribution function $\tilde{F}_i(x)$, which is defined as the cumulative distribution function of negative logarithmic skills of workers from the $i$-th distribution, or $
\tilde{F}_i(t) := P(\tilde{y} \leq t) = P(x \geq \exp(-t)) = 1 - F_i(\exp(-t))$. Using this, we connect $\tilde{y}_i(t)$ to $\underline{t}_i(t)$:
\begin{equation}
\underline{t}_i(t) = 1 - \tilde{F}_i(\tilde{y}_i(t)).  \label{e:linktildeyt}
\end{equation}
which is verified from $\tilde{F}_i(t) = 1 - F_i(\exp(-t))$, together with $\tilde{y}_i(t) = -\log I_i(\underline{t}_i(t))$.

With these new functions, we write the system of functional equations (\ref{eq:lambda_definition}) and (\ref{eq:t_underbar_sum}) as:
\begin{align*}\label{eq:modified_system}
    t & = \sum_{i = 1}^{n} \big(1 - \tilde{F}_i(\tilde{y}_i(t))\big) \\
    t & = \tilde{F}_i \big( \tilde{y}_i(t) - \lambda(t)\big) + \big(1 - \tilde{F}_i(\tilde{y}_i(t))\big) \notag
\end{align*}
for all $i$.\footnote{The second equation is derived from equation (\ref{eq:lambda_definition}) as follows $\lambda(t) = \log I_i(1 - \bar{t}_i(t)) - \log I_i(\underline{t}_i(t)) = \log I_i(1 - t + \underline{t}_i(t)) - \log I_i(\underline{t}_i(t))$ where the second equality follows by (\ref{eq:t_bar_equation}). We simplify the right-hand expression using (\ref{e:ytildedef}) as $\tilde{y}_i(t) - \lambda(t) = -\log I_i(1 - t + \underline{t}_i(t))$ which is equivalent to $1 - t + \underline{t}_i(t) = 1 - \tilde{F}_i(\tilde{y}_i(t) - \lambda(t))$. Reorganizing this equality using (\ref{e:linktildeyt}), we obtain $t = \tilde{F}_i\left( \tilde{y}_i(t) - \lambda(t)\right) + \underline{t}_i(t) = \tilde{F}_i\left( \tilde{y}_i(t) - \lambda(t)\right) + (1 - \tilde{F}_i(\tilde{y}_i(t)))$. Suppose that all distribution functions $\tilde{F}_i$ are concave on $[0, \infty)$, then we can use Lemma 3.1 in \citet{Jakobsons:2016} to show the system of equations has a unique solution. Since we have a symmetric system of functional equations with a unique solution the solution to the system of functional equations is itself symmetric. Uniqueness of the solution thus implies $( \underline{t}_i, \bar{t}_i) $ is identical for all workers $i$. As a result, the marginal worker product $m_i$ in equation (\ref{e:m}) is identical for all worker distributions, or $m_i = m$.}





Given a unique solution, we write the conditions of Proposition \ref{e:propbbapp} in terms of $\tilde{y}_i(t)$. First, we check that $\{ \underline{t}_i \}$ are increasing functions. First, the functions $\underline{t}_i(t) = 1 - \tilde{F}_i(\tilde{y}_i(t))$ is non-negative if and only if $\tilde{F}_i(\tilde{y}_i(t)) \leq 1$. Next, $\underline{t}_i(t)$ is increasing if and only if $\tilde{F}_i(\tilde{y}_i(t))$ is decreasing. By the monotonicity of $\tilde{F}_i$ this is equivalent to $\tilde{y}_i(t)$ is decreasing. We verify that $\tilde{y}_i(t)$ decreases in Lemma \ref{l:propertiesfe}.

The concavity of cumulative distribution functions $\tilde{F}_i$ on $[0, \infty)$ is equivalent to $x_if_i(x_i)$ is non-decreasing on $[0, 1]$, where $f_i$ is the density of the skill distribution. Moreover, it is equivalent to $\tilde{f}_i$, the density function of the negative logarithmic skill distribution, is non-increasing on $[0, \infty)$. The distribution function for the negative logarithmic skill $\tilde{x}_i$ is related to the distribution for $x_i$ through $\tilde{F}_i(\tilde{x}_i) = 1 - F_i(\exp(-\tilde{x}_i))$. The distribution function $\tilde{F}_i$ then has a corresponding density function:
\begin{equation}
\tilde{f}_i(\tilde{x}_i) = \frac{\text{d}}{\text{d} \tilde{x}_i} \tilde{F}_i(\tilde{x}_i) = -  \frac{\text{d}}{\text{d} \tilde{x}_i} F_i(\exp( - \tilde{x}_i)) = f_i (\exp( - \tilde{x}_i)) \exp( - \tilde{x}_i) .
\end{equation}
So, $\tilde{f}_i$ is non-increasing is equivalent to $x_i f_i(x_i)$ is non-decreasing. Therefore, to prove our results, we use Assumption \ref{assumption1}.


Conditions 1 to 3 directly follow since $\{ \tilde{y}_i \}$ is a solution to the system of functional equations (\ref{eq:lambda_definition}) and (\ref{eq:t_underbar_sum}). To verify the fourth condition, we express the function $l(t)$ in terms of $\tilde{y}_i(t)$ using the definition of $\tilde{y}_i(t)$ in (\ref{e:ytildedef}):
\begin{equation}\label{eq:l_using_y}
l(t) = m \lambda(t) - \sum_{i = 1}^{n} \tilde{y}_i(t). 
\end{equation}
The fourth condition holds if and only if the right-hand side of equation (\ref{eq:l_using_y}) is increasing for $0 \leq t \leq q$. We verify this in Lemma \ref{l:propertiesfe}. 

To verify Condition 5, we rewrite the average logarithmic loss for mediocre workers and jobs $h(x)$ in terms of $\tilde{y}_i$:
\begin{align*}
    h(x) = \frac{1}{1 - x}\sum_{i = 1}^{n} \int_{\underline{t}_i(x)}^{1 - \bar{t}_i(x)} \log I_i(t)\,\text{d}t = -\frac{1}{1 - x}\sum_{i = 1}^{n} \int_{\tilde{y}_i(x) - \lambda(x)}^{\tilde{y}_i(x)}z_i\,\text{d} \tilde{F}_i(z_i). 
\end{align*}
The equality follows by a change of variables as $z_i = - \log I_i(t)$, which by definition is distributed according to $\tilde{F}_i$. By equation (\ref{e:ytildedef}) it follows that $-\log I_i(\underline{t}_i(x)) = \tilde{y}_i(x)$ and by the third property of Proposition \ref{e:propbbapp} it follows that $- \log I_i(1 - \bar{t}_i(x)) = - \log I_i (\underline{t}_i(x)) - \lambda(x) = \tilde{y}_i(x) - \lambda(x)$, thereby giving the new bounds of integration. In sum, the fifth condition, $l(q) = h(q)$, holds if and only if 
\begin{equation}\label{e:condition5}
-m \lambda(q) + \sum_{i = 1}^{n}\tilde{y}_i(q) = \frac{1}{1 - q}\sum_{i = 1}^n\int_{\tilde{y}_i(q) - \lambda(q)}^{\tilde{y}_i(q)}z\,\text{d}\tilde{F}_i(z).
\end{equation}

\noindent Lemma 3.2 in \citet{Jakobsons:2016} states the following:


\begin{lemma}\label{l:propertiesfe}
Suppose all cumulative distribution functions $\tilde{F}_i$ are concave on $[0, \infty)$. Let $\{ \tilde{y}_i \}$ and $\lambda$ solve the system of functional equations (\ref{eq:lambda_definition}) and (\ref{eq:t_underbar_sum}). Denote
\begin{equation}
q = \inf\{q \in [0, 1] \colon l(q) \ge h(q)\} \label{e:qcutoof}
\end{equation}
Then on the interval $(0, q)$ for each $i$:
\begin{itemize}[noitemsep]
    \item $0 < \lambda(t) < \tilde{y}_i(t)$
    \item $\lambda(t)$ and $\tilde{y}_i(t)$ are decreasing
    \item $\tilde{y}_i(t) - \lambda(t)$ is increasing
    \item $l(t)$ is increasing
\end{itemize}
\end{lemma}



\noindent We use Lemma \ref{l:propertiesfe} to establish that the fourth and fifth condition are satisfied. Condition 4 follows from the fourth bullet of Lemma \ref{l:propertiesfe}. To verify the fifth condition, we consider three different cases. First, suppose the threshold $q$ constructed in Lemma \ref{l:propertiesfe} is neither equal to $0$ nor equal to $1$. Then, it follows from continuity of $l$ and $h$ that $l(q) = h(q)$ by the definition of $q$, and hence the fifth condition holds automatically. The second case, where $q = 1$, is equivalent to the absence of the mixed set, and hence we do not need to verify Condition 5 and 6 of Proposition \ref{e:propbbapp}. Finally, the case $q = 0$ is equivalent to the absence of countermonotonic sets. In this case, we only need to check Condition 6.

To demonstrate Condition 6 we prove the existence of an assignment such that the output loss is constant across teams. This loss is represented by:
\begin{equation}
x_1 \dots x_n = \exp\Big(-\sum_{i=1}^{n} y_i\Big)
\end{equation}
where $y_i = -\log x_i$ is the negative logarithmic skill of a worker. We can hence equivalently state the mixability condition as there exists an assignment of negative logarithmic skills such that the sum of logarithmic skills across teams remains constant.

To determine the distribution of negative logarithmic skills for mediocre workers, consider the following. The skill distribution is $F_i$ confined to the interval $[I_i(\underline{t}_i(q)), I_i(1 - \bar{t}_i(q))]$. The distribution of negative logarithmic skills is $\tilde{F}_i$; the distribution of negative logarithmic skills is thus $\tilde{F}_i$ restricted to the interval $[- \log I_i(1 - \bar{t}_i(q)), -\log I_i(\underline{t}_i(q))]$. Using equations (\ref{eq:lambda_definition}) and (\ref{e:ytildedef}), $-\log I_i (\underline{t}_i(q)) = \tilde{y}_i(q)$ as well as $-\log I_i(1-\bar{t}_i(x)) = \tilde{y}_i(x)-\lambda(x)$. Thus, the distribution of negative logarithmic skills $\tilde{F}_i$ restricted to the interval \([\tilde{y}_i(q) - \lambda(q), \tilde{y}_i(q)]\). Finally, the total mass of mediocre workers is the same between all distributions and is equal to \(1-q\). We normalize all truncated distributions by this constant to make them probabilistic.

Given the distribution of negative logarithmic skills and the normalization of the truncated distribution, we can write the mixability condition in the following lemma.

\begin{lemma}
The mixability condition, Condition 6 of Proposition \ref{e:propbbapp}, is satisfied if and only if there exists an assignment, denoted \(\tilde{\pi}\), such that:
\begin{itemize}[noitemsep]
    \item The marginal projections of assignment $\tilde{\pi}$ are \(\frac{1}{1-q}\tilde{F}_i\), which are supported on the intervals \([\tilde{y}_i(q) - \lambda(q), \tilde{y}_i(q)]\).
    \item This assignment is supported on the hyperplane defined by \(y_1 + \dots + y_{n} = C\), where \(C\) is a constant representing the sum of negative logarithmic skills across teams.
\end{itemize}
\end{lemma}

\noindent We show mixability under Assumption \ref{assumption1} as discussed above. This assumption implies that each cumulative distribution function $\tilde{F}_i$ of negative logarithmic skills is concave over the interval \([0, \infty)\). Given that the density function of a distribution is the derivative of its cumulative distribution function, Assumption \ref{assumption1} means that the density function of negative log skills is non-increasing.

A direct result is that each distribution of negative log skills is supported on the interval $[\tilde{y}_i(q) - \lambda(q), \tilde{y}_i(q)]$, and that its density function decreases within this interval. Building on this, we introduce Theorem 3.2 of \citet{Wang:2016} which provides necessary and sufficient conditions for mixability under these constraints.

\begin{theorem}[Theorem 3.2 of \citet{Wang:2016}]
Let $\{F_i\}$ be a collection of probability distributions. If each distribution $F_i$ is supported on an interval $[l_i, r_i]$ and has a decreasing density function, then there exists an assignment $\tilde{\pi}$ supported on the hyperplane $\tilde{y}_1 + \dots + \tilde{y}_n = C$ with marginal distributions $\{F_i\}$ respectively if and only if for all $i$:
\begin{equation*}
(r_i - l_i) + \sum_{j = 1}^n l_j \leq \sum_{j = 1}^n \int_{l_j}^{r_j} x_j \text{d} F_j.
\end{equation*}
\end{theorem}

For our problem, this result implies that the distributions of negative logarithmic skills of mediocre workers are mixable if and only if the following inequality is satisfied for each $i$:
\begin{equation}
\lambda(q) + \sum_{i = 1}^n \big( \tilde{y}_i(q) - \lambda(q) \big) \leq \frac{1}{1 - q}\sum_{i = 1}^n \int_{\tilde{y}_i(q) - \lambda(q)}^{\tilde{y}_i(q)} z_i \text{d} \tilde{F}_i(z_i), 
\end{equation}
since $l_k = \tilde{y}_k(q) - \lambda(q)$, $r_k = \tilde{y}_k(q)$, and $F_k = \frac{\tilde{F}_k}{1 - q}$ in our environment. By equation (\ref{e:condition5}) we find that this equation holds with equality if $q\in(0,1)$. Under Condition 5 and Assumption \ref{assumption1}, all conditions are thus met. When $q=0$, $l(q) \geq h(q)$ by the definition of $q$ in (\ref{e:qcutoof}). Hence, the distributions of negative logarithmic skills of mediocre workers are mixable.

\subsection{General Production Functions} \label{a:appprod}

We generalize the functional form of the production function.

\vspace{0.5 cm}
\noindent \textbf{Generalizing Technology for $n=3$}. In the case of two workers, the technology is generalized to additionally incorporate interaction terms between workers and coworkers $x_{1}x_{2}$, workers and firms $x_{1}z$, and coworkers and firms $x_{2}z$. Let the technology be given by:
\begin{equation}
y(x_{1},x_{2},z)=\varphi_{0}+\varphi_{1}x_{1}+\varphi_{2}x_{2}+\varphi_{3}z-\varphi_{12}x_{1}x_{2}-\varphi_{13}x_{1}z-\varphi_{23}x_{2}z-x_{1}x_{2}z,\label{eq:your_equation_label-1}
\end{equation}
where $\varphi_{ij}$ are non-negative to ensure submodularity.

\begin{proposition} Let the technology $y(x_{1},x_{2},z)$
take the form (\ref{eq:your_equation_label-1}). Suppose the
distributions of workers and projects, $F_{x}$ and $F_{z}$ respectively,
satisfy Assumption 1 whereas $\tilde{x}f_{\tilde{x}}\left(\tilde{x}\right)$
is non-decreasing for $\tilde{x}\in\{x+\varphi_{13},x+\varphi_{23},z+\varphi_{12}\}$.
Then  a mixed and countermonotonic assignment  is optimal. \end{proposition}

\begin{proof}
Without loss, we disregard the constant and linear
terms in the technology $y(x_{1},x_{2},z)$, as their contribution to the output loss is independent of the assignment. Consequently, we
express the technology as: 
\begin{equation*}
\hat{y}(x_{1},x_{2},z)=-(x_{1}+\varphi_{23})(x_{2}+\varphi_{13})(z+\varphi_{12}).
\end{equation*}
Let $\hat{z}=z+\varphi_{12}$, $\hat{x}_{1}=x_{1}+\varphi_{23}$,
and $\hat{x}_{2}=x_{2}+\varphi_{13}$. With these substitutions, the
technology function takes its original form $\hat{y}(\hat{x}_{1},\hat{x}_{2},\hat{z})=-\hat{x}_{1}\hat{x}_{2}\hat{z}$.

Given that all $\varphi_{ij}\geq0$, the shifted variables remain
non-negative and are supported on the intervals $[\varphi_{ij},1+\varphi_{ij}]$.\footnote{Although the distributions of $\hat{x}_{1}$, $\hat{x}_{2}$, and
$\hat{z}$ are no longer confined to the unit interval, this does
not impede our proof of existence and optimality. Our proofs neither
leverage the specific scaling of these variables within the unit interval,
nor depend on it in any essential way. We only require each variable
to be non-negative, a condition that is preserved under the transformation.} Denote the distributions corresponding to these variables respectively as $\hat{F}_{1}$,
$\hat{F}_{2}$, and $\hat{F}_{z}$. These distributions satisfy Assumption
1 if and only if the density functions $f_{x}$ and $f_{z}$ of the
original skills satisfy the proposition's stated assumptions. Therefore,
invoking \Cref{proof:support2} to \Cref{a:existencemc}, a mixed and countermonotonic assignment
$\hat{\pi}\in\Pi(\hat{F}_{1},\hat{F}_{2},\hat{F}_{z})$ exists
and solves the assignment problem in this transformed space. The
push-forward image of this assignment $\pi$ in the space of variables
$(x_{1},x_{2},z)$ provides a solution to the initial problem, completing
the proof.
\end{proof}

\noindent The argument in the proof also implies the following corollary.

\begin{corollary}
Consider technology $y(x_{1},x_{2},x_{1}x_{2},x_{2}z,x_{1}z,x_{1}x_{2}z)$.
Suppose its first-order Taylor expansion around the arguments of the
function has negative coefficients on each interaction term
$(x_{1}x_{2},x_{2}z,x_{1}z,x_{1}x_{2}z)$ to ensure submodularity.
Then  a mixed and countermonotonic assignment is optimal.
\end{corollary}

\noindent In sum, our results apply to general functions
of interactions (single, pairwise, triple) in the first-order sense.
The main insight that allows this generalization is the change of
variables.

\vspace{0.5 cm}
\noindent \textbf{Generalizing Technology with $n$ Distributions}. With $n$ distributions, a somewhat different result applies.\footnote{The change of variables argument does not extend when
there are multiple interactions of order higher than two (such as $x_{1}x_{2}z$ or $x_{1}x_{2}x_{3}z$).}

Consider a general technology $y(x_{1},\dots,x_{m},z)$
with heterogeneous workers $\{x_{1},\ldots,x_{m}\}$. For
ease of notation, let the project value $z$ be denoted as $z=x_{m+1}=x_{n}$. Consider a second-order Taylor expansion of this technology,
which is a second-order polynomial in $x_{1},\ldots,x_{n}$, that
can be expressed as: 
\begin{equation}
y(x_{1},\ldots,x_{n})=a_{0}+\sum_{i=1}^{n}a_{i}x_{i}+\sum_{i=1}^{n}\sum_{j=1}^{n}a_{ij}x_{i}x_{j}.\label{e:polynomial-1e}
\end{equation}
Specifically, consider the case where the matrix $A=\{a_{ij}\}$ has rank one and
is elementwise non-positive to ensure submodularity.\footnote{Since the matrix has rank one, each
element $a_{ij}$ can be represented as $a_{ij}=-\lambda_{i}\lambda_{j},$
where $\lambda_{i}\ge0$ for all $i$.}

\begin{proposition} Consider production technology (\ref{e:polynomial-1e}) with $\lambda_{i}\geq0$ and $f_i(x_i)$ is non-increasing. Define transformed skills $\hat{x}_{i}=\exp(-\lambda_{i}x_{i})$,
with $\hat{F}_{x_{i}}$ denoting the corresponding distributions of
$\hat{x}_{i}$. Then, there exists a mixed and countermonotonic assignment
$\hat{\pi}\in\Pi(\hat{F}_{x_{1}},\ldots,\hat{F}_{x_{n}})$ such
that its push-forward image in the space of original skills $(x_{1},\ldots,x_{n})$
is an optimal solution. \end{proposition}

\begin{proof}
We disregard the constant and linear terms in technology (\ref{e:polynomial-1e})
as their contribution to the loss is independent of the assignment. Thus, we focus on the technology:
\begin{equation*}
\hat{y}(x_{1},\ldots,x_{n})=-\sum_{i=1}^{n}\sum_{j=1}^{n}\lambda_{i}\lambda_{j}x_{i}x_{j}=-\big(\sum_{i=1}^{n}\lambda_{i}x_{i}\big)^{2}.
\end{equation*}
With $\hat{x}_{i}=\exp(-\lambda_{i}x_{i})$, we write this technology as $\hat{y}(\hat{x}_{1},\ldots,\hat{x}_{n})=-(\sum_{i=1}^{n}\log \hat{x}_{i} )^{2}$.

First, we verify that each distribution $\hat{F}_{x_{i}}$ satisfies Assumption 1. Observe that: 
\begin{equation*}
\hat{f}_{i}(\hat{x}_{i})=\frac{1}{\lambda_{i}\hat{x}_{i}}f_{i}(x_{i}),
\end{equation*}
which implies that the condition $\hat{x}_{i}\hat{f}_{i}(\hat{x}_{i})$
is non-decreasing is equivalent to the density function $f_{i}(x_{i})$
being non-increasing. Thus, by invoking \Cref{pf:branchbowl} and \Cref{a:existencemc}, we deduce the existence
of a mixed and countermonotonic assignment $\hat{\pi}\in\Pi(\hat{F}_{x_{1}},\ldots,\hat{F}_{x_{n}})$.

In the proof of optimality for a mixed and countermonotonic assignment,
we show using the majorization inequality that the distribution
of the sum of logarithmic skills under $\hat{\pi}$, that is, the distribution
of $\sum\log \hat{x}_{i}$, is minimal in the sense of second-order
stochastic dominance compared to any other assignment $\pi$. Formally,
for any convex function $h$ and alternative assignment $\pi$, the
inequality 
\begin{equation*}
\int h\Big(\sum_{i=1}^{n}\log \hat{x}_{i} \Big) \text{d} \hat{\pi} \leq \int h\Big(\sum_{i=1}^{n}\log \hat{x}_{i} \Big) \text{d} \pi
\end{equation*}
holds. Taking $h(t)=t^{2}$, we find that $\hat{\pi}$ maximizes
total output of the transformed technology $\hat{y}(\hat{x}_{1},\dots,\hat{x}_{n})=-(\sum \log \hat{x}_{i} )^{2}$, and the push-forward of assignment $\hat{\pi}$ to the space of original skills
$(x_{1},\dots,x_{n})$ solves the initial sorting problem.\end{proof}

\subsection{Dual Problem} \label{s:duality_kellerer}

To characterize the solution to the dual problem, we use a generalization of Kantorovich duality for multiple marginal distributions \citep{Kellerer:1984}.\footnote{An instructive exposition of \citet{Kellerer:1984} is presented in \citet{Rachev:1998}.} The duality result states that the optimal value for the planning problem coincides with the optimal value for the dual problem, and that the dual solution exists in the class of integrable functions. 

\vspace{0.5 cm}
\noindent \textbf{Duality}. Let $X$ and $Z$ be compact subsets of the set of real numbers, and $F_{x}$ and $F_z$ be probability measures on $X$ and $Z$ respectively. Let loss function $\ell(x_1, \dots, x_m, z)$ be continuous on $X \times \dots \times X \times Z$ and $\Pi(F_x,\dots,F_x,F_z)$ denote the set of measures $\pi$ on $X \times \dots \times X \times Z$ with fixed projections $\text{pr}_{x_i}(\pi) = F_x$ and $\text{pr}_{z}(\pi) = F_z$. Then,
\begin{equation}
\min_{\pi \in \Pi} \int \ell(x_1, \dots, x_m, z) \text{d} \pi = \max \sum\limits_{i=1}^m \int \hat{w}_i(x_i) \text{d} F_x + \int \hat{v}(z) \text{d} F_z,
\end{equation}
where the supremum is taken with respect to the collection of integrable functions $\hat{w}_i \in L^1(X, F_x)$ and $v \in L^1(Z, F_z)$ subject to the constraint $\sum\limits_{i=1}^m \hat{w}_i(x_i) + \hat{v}(z) \leq \ell(x_1, \dots, x_m, z)$ for any triplet $(x_1, \dots, x_m, z) \in X \times \dots \times X \times Z$.

\vspace{0.5 cm}
\noindent Applying the duality theorem to our problem, we establish that the dual solution attains the same value as the planning problem and that the dual solution exists in the class of integrable functions. We establish that the solution exists in the class of continuous functions. 

\vspace{0.05 cm}
\begin{lemma}
For any dual solution $(w_1,\dots,w_m,v)$, there exist continuous functions $(\hat{w}_1,\dots,\hat{w}_m,\hat{v})$ such that $\hat{w}_i = w_i$ and $\hat{v} = v$ almost everywhere that solves the dual problem. 
\end{lemma}

\begin{proof}
For any dual solution $(w_1,\dots,w_m,v)$ we define:
\begin{equation}
\hat{w}_1(x_1) := \sup_{x_2,\dots,x_m,z} \; y(x_1,\dots,x_m,z) - \sum\limits_{i=2}^m w_i(x_i) - v(z) , \label{e:uniqueness_step}
\end{equation}
where $\hat{w}_1(x_1) \leq w_1(x_1)$ since $w_1(x_1) \geq y(x_1, \dots, x_m, z) - \sum_{i=2}^m w_i(x_i) - v(z)$. By equation (\ref{e:uniqueness_step}), we also have $\hat{w}_1(x_1) \geq y(x_1, \dots, x_m, z) - \sum_{i=2}^m w_i(x_i) - v(z)$, or equivalently, $\hat{w}_1(x_1) + \sum_{i=2}^m w_i(x_i) + v(z) \geq y(x_1, \dots, x_m, z)$ so that $(\hat{w}_1, w_2, \dots, w_n, v)$ is indeed a feasible solution to the dual problem (\ref{e:pp_dual}). Since through $\hat{w}_1(x_1) \leq w_1(x_1)$, it follows that 
\begin{equation*}
\int \hat{w}_1(x_1) \text{d} F_x + \sum\limits_{i=2}^m w_i(x_i) \text{d} F_x  + \int v(z) \text{d} F_z  \leq \int w_1(x_1) \text{d} F_x + \sum\limits_{i=2}^m w_i(x_i) \text{d} F_x + \int v(z) \text{d} F_z.
\end{equation*}
Thus, $(\hat{w}_1,w_2, \dots, w_m, v)$ is a dual solution. Furthermore, $\hat{w}_1 = w_1$ almost everywhere. Suppose not, then instead $\hat{w}_1 < w_1$, which would imply that:
\begin{equation*}
\int \hat{w}_1(x_1) \text{d} F_x + \sum\limits_{i=2}^m w_i(x_i) \text{d} F_x  + \int v(z) \text{d} F_z  < \int w_1(x_1) \text{d} F_x + \sum\limits_{i=2}^m w_i(x_i) \text{d} F_x + \int v(z) \text{d} F_z.
\end{equation*}
which would contradict that $(w_1, \dots, w_m, v)$ is a dual solution. 

We next establish that $\hat{w}_1$ is a continuous function. To show this, let $\hat{w}_1(x_1 ; x_2, \dots, x_m, z) := y(x_1, \dots, x_m, z) - \sum_{i=2}^m w_i(x_i) - v(z)$, that is, (\ref{e:uniqueness_step}) holding constant all coworkers $(x_2, \dots, x_m)$ and the firm $z$. Given the technology (\ref{yp}), and since $x_i \in [0,1]$ and $z \in [0,1]$, $\hat{w}_1(x_1 ; x_2, \dots, x_m, z)$ is a linear function in $x_1$ with slope less than one in absolute value. This implies $\hat{w}_1(x_1 ; x_2, \dots, x_m, z)$ is Lipschitz continuous. Since $\hat{w}_1(x_1)$ can be represented as $\hat{w}_1(x_1) = \sup \limits_{x_2, \dots, x_m, z} \hat{w}_1(x_1 ; x_2, \dots, x_m, z)$, $\hat{w}_1(x_1)$ is Lipschitz continuous as the supremum of uniformly Lipschitz continuous functions is Lipschitz continuous. Hence, $\hat{w}_1(x_1)$ is continuous. 


Similarly, we sequentially define $\{ \hat{w}_i \}$ and $\hat{v}$: 
\begin{align*}
\hat{w}_i(x_i) & := \hspace{0.22 cm} \sup_{x_{j \neq i}, z} \hspace{0.13 cm} y(x_1, \dots, z) - \sum\limits_{j=1}^{i-1} \hat{w}_j(x_j) - \sum\limits_{j=i+1}^m w_j(x_j) - v(z),\\
\hat{v}(z) & := \sup_{x_1, \dots, x_m} y(x_1, \dots, z) - \sum\limits_{i=1}^m \hat{w}_i(x_i).
\end{align*}
By sequentially applying the above argument, we obtain that $(\hat{w}_1, \dots, \hat{w}_m, \hat{v})$ is a dual solution with $\hat{w}_i = w_i$ for all $i$ and $\hat{v} = v$ almost everywhere. Similarly, the same arguments shows that all $\hat{w}_i$ and $\hat{v}$ are continuous.\end{proof} 

\subsection{Proposition \ref{prop:cont_derivative}} \label{proof:cont_derivative}

For any worker $x_{1,0}$ there exist coworkers $\{ x_{j,0} \}_{j=2}^m$ and a firm $z_0$ so that team $(x_{1,0}, \dots, x_{m, 0}, z_0)$ is contained in the support of an optimal assignment, implying that:
\begin{equation*}
w(x_{1,0}) + \dots + w(x_{m,0}) + v(z_{0}) = y(x_{1,0}, \dots, x_{m,0}, z_{0}) = z_{0} + x_{1,0} m_x(x_{1,0}) ,
\end{equation*}
with the second equality following from the marginal worker product (\ref{mwp}). For any other worker $x_1 $, it must by the constraint be that $w(x_{1}) + \dots + w(x_{m,0}) + v(z_{0}) \geq z_{0} + x_{1} m(x_{1,0})$. Combining these expressions, it holds that for any $x_0$ and $x_1 \in [0,1]$:
\begin{equation}
w(x_1) - w(x_{1,0}) \geq m(x_{1,0}) ( x_{1} -  x_{1,0} ) . \label{e:1w}
\end{equation}

Similarly, by interchanging $x_{1,0}$ and $x_1$ in the previous paragraph, we obtain:
\begin{equation}
w(x_{1,0}) - w(x_{1}) \geq m(x_{1}) ( x_{1,0} -  x_{1} ) . \label{e:2w}
\end{equation}
Combining (\ref{e:1w}) and (\ref{e:2w}), dividing by $(x_{1}-x_{1,0})$ yields
\begin{equation}
m(x_{1}) \geq \frac{w(x_{1}) - w(x_{1,0})}{x_{1}-x_{1,0}} \geq  m(x_{1,0}).
\end{equation}
Taking limits as $x_1 \rightarrow x_{1,0}$, this equation shows that the function $w(x)$ has a derivative at $x_1 = x_{1,0}$ and that this derivative is equal to $m(x_{1,0})$. Since the marginal worker product $m$ is a continuous function by Proposition \ref{prop:marginal_product}, the derivative of $w$ is continuous. The result is similarly established for $v'(z)$.

\subsection{Quantitative Robustness and Extensions} \label{a:appquant}

The quantitative analysis in the main text assumes that teams are of size three – two workers and a project. In this appendix, we demonstrate that our model can also account for the decomposition of earnings variation within and across firms with teams of sizes four, five, and six. We also find that the non-targeted average coworker earnings are robust to team size. As the number of team members increases, the firm and worker effect become more muted. These results are summarized in \Cref{p:three}, \Cref{p:four}, and \Cref{p:five}.

\begin{table}
\def\arraystretch{1.95}
\begin{center}
\caption{Parametric Identification with Three Workers}\label{p:three}
\begin{tabular}{l|ccc|ccc}
\hline  \hline
 \multicolumn{7}{c}{Model and Data Earnings Decomposition} \\ \hline 
\hspace{1.5 cm} & \multicolumn{3}{c|}{Data}     & \multicolumn{3}{c}{Model}     \\
Moment  \hspace{0.1 cm} & \hspace{0.1 cm} 1981 \hspace{0.1 cm}  & \hspace{0.1 cm} 2013 \hspace{0.1 cm} & \hspace{0.1 cm} change \hspace{0.1 cm} & \hspace{0.1 cm} 1981 \hspace{0.1 cm} & \hspace{0.1 cm} 2013 \hspace{0.1 cm} & \hspace{0.1 cm} change \hspace{0.1 cm} \\ \hline
Between & 0.34     & 0.42 & \phantom{-}0.08 & 0.34     & 0.42 & \phantom{-}0.08 \\
Within & 0.66     & 0.58 & -0.08 & 0.66     & 0.58 & -0.08 \\ 
 \hline \hline
 \multicolumn{7}{c}{Model and Data Coworker Earning} \\ \hline 
\hspace{1.5 cm} & \multicolumn{3}{c|}{Data}     & \multicolumn{3}{c}{Model}     \\
Percentile  \hspace{0.1 cm} & \hspace{0.1 cm} 1981 \hspace{0.1 cm}  & \hspace{0.1 cm} 2013 \hspace{0.1 cm} & \hspace{0.1 cm} change \hspace{0.1 cm} & \hspace{0.1 cm} 1981 \hspace{0.1 cm} & \hspace{0.1 cm} 2013 \hspace{0.1 cm} & \hspace{0.1 cm} change \hspace{0.1 cm} \\ \hline
25 & 10.01     & 10.07 & 0.06 & 10.19    & 10.28 & 0.09 \\
50 & 10.30     & 10.45 & 0.15 & 10.44     & 10.61 & 0.17 \\
75 & 10.56     & 10.75 & 0.19 & 10.56     & 10.79 & 0.23 \\ 
90 & 10.67     & 10.98 & 0.31 & 10.55     & 10.82 & 0.27 \\ 
 \hline \hline
\end{tabular}
\begin{tabular}{l|cc|cc|cc}
\multicolumn{7}{c}{Model and Data Earnings Decomposition} \\ \hline 
\hspace{1.5 cm} & \multicolumn{2}{c|}{Model}     & \multicolumn{2}{c|}{Firm Effect} & \multicolumn{2}{c}{Worker Effect} \\
Moment  \hspace{0.1 cm} & \hspace{0.1 cm} 1981 \hspace{0.1 cm}  & \hspace{0.1 cm} 2013 \hspace{0.1 cm} & \hspace{0.1 cm} 2013 \hspace{0.1 cm} & \hspace{0.1 cm} change \hspace{0.1 cm} & \hspace{0.1 cm} 2013 \hspace{0.1 cm} & \hspace{0.1 cm} change \hspace{0.1 cm} \\ \hline
Between & 0.34     & 0.42   & 0.24 & -0.10  & 0.52 & \phantom{-}0.18 \\
Within & 0.66    & 0.58      & 0.76 & \phantom{-}0.10  & 0.48 & -0.18 \\ 
 \hline \hline
\end{tabular}
\end{center}
\end{table}

\begin{table}
\def\arraystretch{1.95}
\begin{center}
\caption{Parametric Identification with Four Workers}\label{p:four}
\begin{tabular}{l|ccc|ccc}
\hline  \hline
 \multicolumn{7}{c}{Model and Data Earnings Decomposition} \\ \hline 
\hspace{1.5 cm} & \multicolumn{3}{c|}{Data}     & \multicolumn{3}{c}{Model}     \\
Moment  \hspace{0.1 cm} & \hspace{0.1 cm} 1981 \hspace{0.1 cm}  & \hspace{0.1 cm} 2013 \hspace{0.1 cm} & \hspace{0.1 cm} change \hspace{0.1 cm} & \hspace{0.1 cm} 1981 \hspace{0.1 cm} & \hspace{0.1 cm} 2013 \hspace{0.1 cm} & \hspace{0.1 cm} change \hspace{0.1 cm} \\ \hline
Between & 0.34     & 0.42 & \phantom{-}0.08 & 0.34     & 0.42 & \phantom{-}0.08 \\
Within & 0.66     & 0.58 & -0.08 & 0.66     & 0.58 & -0.08 \\ 
 \hline \hline
 \multicolumn{7}{c}{Model and Data Coworker Earning} \\ \hline 
\hspace{1.5 cm} & \multicolumn{3}{c|}{Data}     & \multicolumn{3}{c}{Model}     \\
Percentile  \hspace{0.1 cm} & \hspace{0.1 cm} 1981 \hspace{0.1 cm}  & \hspace{0.1 cm} 2013 \hspace{0.1 cm} & \hspace{0.1 cm} change \hspace{0.1 cm} & \hspace{0.1 cm} 1981 \hspace{0.1 cm} & \hspace{0.1 cm} 2013 \hspace{0.1 cm} & \hspace{0.1 cm} change \hspace{0.1 cm} \\ \hline
25 & 10.01     & 10.07 & 0.06 & 10.24    & 10.28 & 0.04 \\
50 & 10.30     & 10.45 & 0.15 & 10.45     & 10.62 & 0.17 \\
75 & 10.56     & 10.75 & 0.20 & 10.52     & 10.77 & 0.25 \\ 
90 & 10.67     & 10.98 & 0.31 & 10.50     & 10.79 & 0.29 \\ 
 \hline \hline
\end{tabular}
\begin{tabular}{l|cc|cc|cc}
\multicolumn{7}{c}{Model and Data Earnings Decomposition} \\ \hline 
\hspace{1.5 cm} & \multicolumn{2}{c|}{Model}     & \multicolumn{2}{c|}{Firm Effect} & \multicolumn{2}{c}{Worker Effect} \\
Moment  \hspace{0.1 cm} & \hspace{0.1 cm} 1981 \hspace{0.1 cm}  & \hspace{0.1 cm} 2013 \hspace{0.1 cm} & \hspace{0.1 cm} 2013 \hspace{0.1 cm} & \hspace{0.1 cm} change \hspace{0.1 cm} & \hspace{0.1 cm} 2013 \hspace{0.1 cm} & \hspace{0.1 cm} change \hspace{0.1 cm} \\ \hline
Between & 0.34     & 0.42   & 0.22 & -0.12  & 0.49 & \phantom{-}0.15 \\
Within & 0.66    & 0.58      & 0.78 & \phantom{-}0.12  & 0.51 & -0.15 \\ 
 \hline \hline
\end{tabular}
\end{center}
\end{table}

\begin{table}
\def\arraystretch{1.95}
\begin{center}
\caption{Parametric Identification with Five Workers}\label{p:five}
\begin{tabular}{l|ccc|ccc}
\hline  \hline
 \multicolumn{7}{c}{Model and Data Earnings Decomposition} \\ \hline 
\hspace{1.5 cm} & \multicolumn{3}{c|}{Data}     & \multicolumn{3}{c}{Model}     \\
Moment  \hspace{0.1 cm} & \hspace{0.1 cm} 1981 \hspace{0.1 cm}  & \hspace{0.1 cm} 2013 \hspace{0.1 cm} & \hspace{0.1 cm} change \hspace{0.1 cm} & \hspace{0.1 cm} 1981 \hspace{0.1 cm} & \hspace{0.1 cm} 2013 \hspace{0.1 cm} & \hspace{0.1 cm} change \hspace{0.1 cm} \\ \hline
Between & 0.34     & 0.42 & \phantom{-}0.08 & 0.34     & 0.42 & \phantom{-}0.08 \\
Within & 0.66     & 0.58 & -0.08 & 0.66     & 0.58 & -0.08 \\ 
 \hline \hline
 \multicolumn{7}{c}{Model and Data Coworker Earning} \\ \hline 
\hspace{1.5 cm} & \multicolumn{3}{c|}{Data}     & \multicolumn{3}{c}{Model}     \\
Percentile  \hspace{0.1 cm} & \hspace{0.1 cm} 1981 \hspace{0.1 cm}  & \hspace{0.1 cm} 2013 \hspace{0.1 cm} & \hspace{0.1 cm} change \hspace{0.1 cm} & \hspace{0.1 cm} 1981 \hspace{0.1 cm} & \hspace{0.1 cm} 2013 \hspace{0.1 cm} & \hspace{0.1 cm} change \hspace{0.1 cm} \\ \hline
25 & 10.01     & 10.07 & 0.06 & 10.20    & 10.28 & 0.08 \\
50 & 10.30     & 10.45 & 0.15 & 10.43     & 10.62 & 0.19 \\
75 & 10.56     & 10.75 & 0.19 & 10.52     & 10.76 & 0.24 \\ 
90 & 10.67     & 10.98 & 0.31 & 10.51     & 10.79 & 0.28 \\ 
 \hline \hline
\end{tabular}
\begin{tabular}{l|cc|cc|cc}
\multicolumn{7}{c}{Model and Data Earnings Decomposition} \\ \hline 
\hspace{1.5 cm} & \multicolumn{2}{c|}{Model}     & \multicolumn{2}{c|}{Firm Effect} & \multicolumn{2}{c}{Worker Effect} \\
Moment  \hspace{0.1 cm} & \hspace{0.1 cm} 1981 \hspace{0.1 cm}  & \hspace{0.1 cm} 2013 \hspace{0.1 cm} & \hspace{0.1 cm} 2013 \hspace{0.1 cm} & \hspace{0.1 cm} change \hspace{0.1 cm} & \hspace{0.1 cm} 2013 \hspace{0.1 cm} & \hspace{0.1 cm} change \hspace{0.1 cm} \\ \hline
Between & 0.34     & 0.42   & 0.38 & \phantom{-}0.04  & 0.38 & \phantom{-}0.04 \\
Within & 0.66    & 0.58      & 0.62 & -0.04  & 0.62 & -0.04 \\ 
 \hline \hline
\end{tabular}
\end{center}
\end{table}

\end{document}